\documentclass{sig-alternate-05-2015}

\usepackage{cite}
\usepackage{graphicx}
\usepackage{caption}
\usepackage{subcaption}
\usepackage{amsthm}
\theoremstyle{plain}

\usepackage{amssymb} 
\usepackage{amsmath} 
\usepackage{mathtools}
\usepackage{mathrsfs}
\usepackage{tikz}
\makeatother

\newtheorem{corr}{Corollary}
\newtheorem{lem}{Lemma}
\newtheorem{rem}{Remark}
\newtheorem{prop}{Proposition}
\newtheorem{prob}{Problem}
\newtheorem{thm}{Theorem}
\newtheorem{innercustomthm}{Problem}
\newenvironment{customthm}[1]
  {\renewcommand\theinnercustomthm{#1}\innercustomthm}
  {\endinnercustomthm}

\usepackage{xspace}
\usepackage{amsmath,amssymb,amsfonts}

\newcommand{\bs}[1]{\ensuremath{\boldsymbol{#1}}}

\newcommand{\ba}{\ensuremath{\bs a}\xspace}

\newcommand{\bv}{\ensuremath{\bs v}\xspace}
\newcommand{\bw}{\ensuremath{\bs w}\xspace}
\newcommand{\bx}{\ensuremath{\bs x}\xspace}

\newcommand{\bz}{\ensuremath{\bs z}\xspace}

\newcommand{\bW}{\ensuremath{\bs W}\xspace}

\newcommand{\bxi}{\ensuremath{\bs \xi}\xspace}

\newcommand{\pci}[1]{\ensuremath{\bs{p}_{\boldsymbol{c}_{i}}}\xspace}

\newcommand{\Prob}{\mathbb{P}}
\newcommand{\Exp}{\mathbb{E}}
\newcommand{\FSRset}{\mathrm{FSReach}}
\newcommand{\probref}[1]{{#1}}

\begin{document}

\CopyrightYear{2017}
\setcopyright{acmcopyright}
\conferenceinfo{HSCC '17,}{April 18--20, 2017, Pittsburgh, PA, USA}
\isbn{978-1-4503-4590-3/17/04}\acmPrice{\$15.00}
\doi{http://dx.doi.org/10.1145/3049797.3049818}

\clubpenalty=10000
\widowpenalty = 10000

\title{Forward Stochastic Reachability Analysis for Uncontrolled Linear Systems using Fourier Transforms}
\numberofauthors{3} 
\author{
\alignauthor 
Abraham P. Vinod\\
    \affaddr{Electrical \& Comp. Eng.}\\
    \affaddr{University of New Mexico}\\
    \affaddr{Albuquerque, NM 87131, USA}\\
    \email{aby.vinod@gmail.com}
\alignauthor 
{Baisravan HomChaudhuri}\\
    \affaddr{Electrical \& Comp. Eng.}\\
    \affaddr{University of New Mexico}\\
    \affaddr{Albuquerque, NM 87131, USA}\\
    \email{baisravan.hc@gmail.com}
\alignauthor 
Meeko M. K. Oishi\\
    \affaddr{Electrical \& Comp. Eng.}\\
    \affaddr{University of New Mexico}\\
    \affaddr{Albuquerque, NM 87131, USA}\\
    \email{oishi@unm.edu}
}
\date{}
\maketitle
\begin{abstract}
We propose a scalable method for forward stochastic reachability analysis for uncontrolled linear systems with affine disturbance. Our method uses Fourier transforms to efficiently compute the forward stochastic reach probability measure (density) and the forward stochastic reach set. This method is applicable to systems with bounded or unbounded disturbance sets.  We also examine the convexity properties of the forward stochastic reach set and its probability density. Motivated by the problem of a robot attempting to capture a stochastically moving, non-adversarial target, we demonstrate our method on two simple examples. Where traditional approaches provide approximations, our method provides exact analytical expressions for the densities and probability of capture.
\end{abstract}

%


\section{Introduction}
\label{sec:introduction}

Reachability analysis of discrete-time dynamical systems with stochastic
disturbance input is an established tool to provide probabilistic assurances of
safety or performance and has been applied in several domains, including motion
planning in robotics~\cite{Baisravan2017ACC,malone2014stochastic}, spacecraft
docking~\cite{lesser_stochastic_2013}, fishery management and mathematical
finance~\cite{summers_verification_2010}, and autonomous
survelliance~\cite{kariotoglou2011stochastic}.  The computation of stochastic
reachable and viable sets has been formulated within a dynamic programming
framework \cite{Abate2008,summers_verification_2010} that generalizes to
discrete-time stochastic hybrid systems, and suffers from the well-known curse
of dimensionality \cite{abate_computational_2007}.  Recent work in computing
stochastic reachable and viable sets aims to circumvent these computational
challenges, through approximate dynamic programming
\cite{kariotoglou2013ECC,Krtgl2016, manganini_policy_2015}, Gaussian mixtures
\cite{Krtgl2016}, particle filters \cite{manganini_policy_2015,
lesser_stochastic_2013}, and convex chance-constrained optimization
\cite{lesser_stochastic_2013, kariotoglou2011stochastic}.  These methods have
been applied to systems that are at most 6-dimensional \cite{kariotoglou2013ECC}
-- far beyond the scope of what is possible with dynamic programming, but are
not scalable to larger and more realistic scenarios.

We focus in particular on the forward stochastic reachable set, defined as the
smallest closed set that covers all the reachable states.  For LTI systems with
bounded disturbances, established verification methods
\cite{kvasnica2015reachability, ellipsoid, girard2005reachability} can be
adapted to overapproximate the forward stochastic reachable set.  However, these
methods return a trivial result with unbounded disturbances and do not address
the forward stochastic reach probability measure, which provides the likelihood
of reaching a given set of states.

We present a scalable method to perform forward stochastic reachability analysis
of LTI systems with stochastic dynamics, that is, a method to compute the
forward stochastic reachable set as well as its probability measure.  We show
that Fourier transforms can be used to provide exact reachability analysis, for
systems with bounded or unbounded disturbances.  We provide both iterative and
analytical expressions for the probability density, and show that explicit
expressions can be derived in some cases.

We are motivated by a particular application: pursuit of a dynamic,
non-adversarial target \cite{hollinger2009efficient}.  Such a scenario may arise
in e.g., the rescue of a lost first responder in a building on fire
\cite{kumar2004robot}, capture of a non-aggressive UAV in an urban
environment~\cite{geyer2008active}, or other non-antagonistic situations.
Solutions for an adversarial target, based in a two-person, zero-sum
differential game, can accommodate bounded disturbances with unknown
stochasticity \cite{mitchell2000level, tomlin2000game,
tomlin_computational_2003,
bokanowski_reachability_2010,Huang2015AutomationAssistedCA}, but will be
conservative for a non-adversarial target. 
We seek scalable solutions that synthesize an optimal
controller for the non-adversarial scenario, by exploiting the forward reachable
set and probability measure for the target.  We analyze the convexity properties
of the forward stochastic reach probability density and sets, and propose a
convex optimization problem to provide the exact probabilistic guarantee of
success and the corresponding optimal controller.

The main contributions of this paper are: 1) a method to efficiently compute the
forward stochastic reach sets and the corresponding probability measure for
linear systems with uncertainty using Fourier transforms, 2) the convexity
properties of the forward stochastic reach probability measure and sets, and 3)
a convex formulation to maximize the probability of capture of a non-adversarial
target with stochastic dynamics using the forward stochastic
reachability analysis. 

The paper is organized as follows: We define the forward stochastic reachability
problem and review some properties from probability theory and Fourier analysis
in Section~\ref{sec:preliminaries}.  Section~\ref{sec:computeReach} formulates
the forward stochastic reachability analysis for linear systems using Fourier
transforms and provides convexity results for the probability measure and the
stochastic reachable set. We apply the proposed method to solve the controller
synthesis problem in Section~\ref{sec:stochTarget}, and provide conclusions and
directions for future work in Section~\ref{sec:conc}.

\section{Preliminaries and Problem Formulation}
\label{sec:preliminaries}

In this section, we review some properties from probability theory and Fourier
analysis relevant for our discussion and setup the problems. For detailed
discussions on probability theory, see~\cite{billingsley_probability_1995,
gubner_probability_2006,cramer2016mathematical,dharmadhikari1988unimodality},
and on Fourier analysis, see~\cite{stein1971introduction}. We denote random
vectors with bold case and non-random vectors with an overline.

\subsection{Preliminaries}
\label{sub:prelim}

A random vector $\bw\in \mathbb{R}^{p}$ is defined in a probability
space  $(\mathcal{W},\sigma( \mathcal{W}), \Prob_{\bw})$.
Given a sample space $ \mathcal{W}$, the sigma-algebra $\sigma( \mathcal{W})$
provides a collection of measurable sets defined over $ \mathcal{W}$. The sample
space can be either countable (discrete random vector $\bw$) or uncountable
(continuous random vector $\bw$). In this paper, we focus only on absolutely
continuous random variables.  For an absolutely continuous random vector, the
probability measure defines a probability density function $\psi_{\bw}:
\mathbb{R}^p \rightarrow \mathbb{R}$ such that given a (Borel) set $
\mathcal{B}\in \sigma(\mathcal{W})$, we have $\Prob_{\bw}\{\bw\in
\mathcal{B}\}=\int_{ \mathcal{B}} \psi_{\bw}(\bar{z})d\bar{z}$. Here, $d\bar{z}$
is short for $dz_1dz_2\ldots dz_p$.  

We will use the concept of support to define the forward stochastic reach set.
The support of a random vector is the smallest closed set that will occur almost
surely. Formally, the support of a random vector $\bw$ is a unique
\emph{minimal} closed set $ \mathrm{supp}(\bw)\in\sigma( \mathcal{W})$ such that
1) $ \Prob_{\bw}\left\{ \bw\in \mathrm{supp}(\bw)\right\}=1$, and 2) if $
\mathcal{D} \in \sigma( \mathcal{W})$ such that $\Prob_{\bw}\left\{  \bw\in
\mathcal{D}\right\}=1$, then $ \mathrm{supp}(\bw)\subseteq
\mathcal{D}$~\cite[Section 10, Ex.  12.9]{billingsley_probability_1995}.
Alternatively, denoting the Euclidean ball of radius $\delta$ centered at
$\bar{z}$ as $ \mathrm{Ball}(\bar{z},\delta)$, we have \eqref{eq:support_def}
which is equivalent to \eqref{eq:support_def2} via~\cite[Proposition
19.3.2]{tao_analysis2},
\begin{align}
    \mathrm{supp}(\bw)&=\left\{\bar{z}\in \mathcal{W}\vert\forall \delta>0,
\int_{
\mathrm{Ball}(\bar{z},\delta)}\psi_{\bw}(\bar{z})d\bar{z}>0\right\}\label{eq:support_def}\\
    &= \mathcal{W}\setminus\left\{\bar{z}\in \mathcal{W}\vert\exists \delta>0,
\psi_{\bw}(\bar{z})=0\mbox{ a.e. in
}\mathrm{Ball}(\bar{z},\delta)\right\}\label{eq:support_def2}
\end{align}
For a continuous $\psi_{\bw}$, \eqref{eq:support_def2} is the support of the
density~\cite[Section 8.8]{AnalysisJean}. Denoting the closure of a set using $
\mathrm{cl}(\cdot)$,
\begin{align}
    \mathrm{supp}(\bw)&=\mathrm{support}(\psi_{\bw})= \mathrm{cl}(\{\bar{z}\in \mathcal{W}\vert
\psi_{\bw}(\bar{z})>0\}).\label{eq:support_def3}
\end{align}

The characteristic function (CF) of a random vector
$\bw\in \mathbb{R}^{p}$ with probability density function
$\psi_{\bw}(\bar{z})$ is
\begin{align}
    \Psi_{\bw}(\bar{\alpha})&\triangleq
    \Exp_{\bw}\left[\mathrm{exp}\left({j\bar{\alpha}^\top\bw}\right)\right] \nonumber \\
                       &=\int_{\mathbb{R}^p}e^{j\bar{\alpha}^\top\bar{z}}
    \psi_{\bw}(\bar{z})d\bar{z}=
    \mathscr{F}\left\{\psi_{\bw}(\cdot)\right\}(-\bar{\alpha})\label{eq:cfun_def}
\end{align}
where $ \mathscr{F}\{\cdot\}$ denotes the Fourier transformation operator and
$\bar{\alpha}\in \mathbb{R}^{p}$. Given a CF
$\Psi_{\bw}(\bar{\alpha})$, the density function can be computed as
\begin{align}
    \psi_{\bw}(\bar{z})&=
    \mathscr{F}^{-1}\left\{\Psi_{\bw}(\cdot)\right\}(-\bar{z})
    \nonumber \\
    &={\left(\frac{1}{2\pi}\right)}^p\int_{
\mathbb{R}^p}e^{-j\bar{\alpha}^\top\bar{z}}
\Psi_{\bw}(\bar{\alpha})d\bar{\alpha}\label{eq:cfun_ift}
\end{align}
where $ \mathscr{F}^{-1}\{\cdot\}$ denotes the inverse Fourier transformation
operator and $d\bar{\alpha}$ is short for $d\alpha_1d\alpha_2\ldots d\alpha_p$.

We define the $L^d( \mathbb{R}^p)$ spaces, $1\leq d< \infty$, of measurable
real-valued functions with finite $L^d$ norm. The $L^d$ norm of a 
density $\psi_{\bw}$ is ${\Vert \psi_{\bw}\Vert}_d\triangleq {\left(\int_{
\mathbb{R}^d}{\vert\psi_{\bw}(\bar{z})\vert}^dd\bar{z}\right)}^{1/d}$ where
$\vert \cdot\vert$ denotes the absolute value. Here, $L^1( \mathbb{R}^p)$ is the
space of absolutely integrable functions, and $L^2( \mathbb{R}^p)$ is the space
of square-integrable functions. 
The Fourier transformation is defined for all functions in
$L^1( \mathbb{R}^p)$ and all functions in $L^2( \mathbb{R}^p)$. 
Since probability densities
are, by definition, in $L^1( \mathbb{R}^p)$, CFs exist for every
probability density~\cite[Section 1]{stein1971introduction}. 
Let $\bw_1,\bw_2\in \mathbb{R}^p$ be random vectors with densities
$\psi_{\bw_1}$ and $\psi_{\bw_2}$ and CFs $\Psi_{\bw_1}$
and $\Psi_{\bw_2}$ respectively. By definition, $\psi_{\bw_1},\psi_{\bw_2}\in
L^1( \mathbb{R}^p)$.  Let
$\bar{z},\bar{z}_1,\bar{z}_2,\bar{\alpha},\bar{\alpha}_1,\bar{\alpha}_2\in
\mathbb{R}^p,\bar{\beta}\in \mathbb{R}^n$.
\begin{enumerate}
    \item[P1)]  If $\bx=\bw_1+\bw_2$, then
        $\psi_{\bx}(\bar{z})=\big(\psi_{\bw_1}(\cdot)\ast\psi_{\bw_2}(\cdot)\big)(\bar{z})$
        and $\Psi_{\bx}(\bar{\alpha}) = \Psi_{\bw_1}(\bar{\alpha})
        \Psi_{\bw_2}(\bar{\alpha})$~\cite[Section
        21.11]{cramer2016mathematical}. Also,
        $ \mathrm{supp}(\bx) \subseteq
        \mathrm{cl}(\mathrm{supp}(\bw_1) \oplus \mathrm{supp}(\bw_2))
        $~\cite[Lemma 8.15]{AnalysisJean}.  Here, $\ast$
        denotes convolution and $\oplus$ Minkowski sum.
    \item[P2)] If $\bx=F\bw_1+G$ where $F\in \mathbb{R}^{p\times n},G\in
        \mathbb{R}^{n}$ are matrices,
        $\Psi_{\bx}(\bar{\beta})=\mathrm{exp}\left({j\bar{\beta}^\top
        G}\right)\Psi_{\bw_1}(F^\top\bar{\beta})$ (from~\cite[Section
        22.6]{cramer2016mathematical} and~\cite[Equation
        1.5]{stein1971introduction}).
    \item[P3)] If $\bw_1$ and $\bw_2$ are independent vectors, then
        $\bx={[\bw_1^\top\ \bw_2^\top]}^\top$ has probability density
        $\psi_{\bx}(\bar{y})=\psi_{\bw_1}(\bar{z}_1)\psi_{\bw_2}(\bar{z}_2),
        \bar{y}={[\bar{z}_1^\top\ \bar{z}_2^\top]}^\top\in \mathbb{R}^{2p}$ and
        CF
        $\Psi_{\bx}(\bar{\gamma})=\Psi_{\bw_1}(\bar{\alpha}_1)\Psi_{\bw_2}(\bar{\alpha}_2),
        \bar{\gamma}={[\bar{\alpha}_1^\top\ \bar{\alpha}_2^\top]}^\top\in
        \mathbb{R}^{2p}$~\cite[Section 22.4]{cramer2016mathematical}.
    \item[P4)] The marginal probability density of any group of $k$ components
        selected from the random vector $\bw_1$ is obtained by setting
        the remaining $p-k$ Fourier variables in
        the CF to zero~\cite[Section
        22.4]{cramer2016mathematical}.
\end{enumerate}

An additional assumption of square-integrability of the probability density of
the random variable $\bw_3\in \mathbb{R}^p$ results in $\psi_{\bw_3}\in L^1(
\mathbb{R}^p)\cap L^2( \mathbb{R}^p)$. Along with Properties P1-P4,
$\psi_{\bw_3}$ satisfies the following property:
\begin{enumerate}
    \item[P5)] The Fourier transform preserves the inner product in $L^2(
        \mathbb{R}^p)$~\cite[Theorem 2.3]{stein1971introduction}. Given a
        square-integrable function $h(\bar{z})$ with Fourier transform
        $H(\bar{\alpha})=\mathscr{F}\left\{h(\cdot)\right\}(\bar{\alpha})$ and a
        square-integrable probability density $\psi_{\bw_3}$,
        \begin{align}
            \int_{ \mathbb{R}^p}{\psi_{\bw_3}(\bar{z})}^\dagger
            h(\bar{z})d\bar{z}={\left(\frac{1}{2\pi}\right)}^p
            \int_{
        \mathbb{R}^p}&{\big(\mathscr{F}\left\{\psi_{\bw_3}(\cdot)\right\}(\bar{\alpha})\big)}^\dagger \nonumber \\
        & \quad \times H(\bar{\alpha})d\bar{\alpha}\nonumber
        \end{align}
        Here, $\dagger$ denotes complex conjugation.
\end{enumerate}
\begin{lem}\label{lem:capturePr}
    For square-integrable $\psi_{\bw_3}$  and $h$,
        \begin{align}
            \int_{ \mathbb{R}^p}\psi_{\bw_3}(\bar{z})
            h(\bar{z})d\bar{z}&={\left(\frac{1}{2\pi}\right)}^p\int_{
        \mathbb{R}^p}\Psi_{\bw_3}(\bar{\alpha})
        H(\bar{\alpha})d\bar{\alpha}.\label{eq:cfun_capturePr}
        \end{align}
\end{lem}
\begin{proof}
Follows from Property P5, \eqref{eq:cfun_ift}, and ~\cite[Section
10.6]{cramer2016mathematical}. Since probability densities are real functions,
${\big(\psi_{\bw_3}(\bar{z})\big)}^\dagger={\psi_{\bw_3}(\bar{z})}$ and
${\big(\mathscr{F}\left\{\psi_{\bw_3}(\cdot)\right\}(\bar{\alpha})\big)}^\dagger
=\Psi_{\bw_3}(\bar{\alpha})$.
\end{proof}


\subsection{Problem formulation}
\label{sub:problem_formulation}

Consider the discrete-time linear time-invariant system,
\begin{align}
    \bx[t+1]&=A\bx[t]+B\bw[t]\label{eq:sys_orig}
\end{align}
with state $\bx[t]\in \mathcal{X}\subseteq \mathbb{R}^n$, disturbance $w[t]\in
\mathcal{W}\subseteq \mathbb{R}^p$, and matrices $A,B$ of appropriate
dimensions. Let $\bar{x}_0\in \mathcal{X}$  be the given initial state and $T$
be the finite time horizon. 
The disturbance set $\mathcal{W}$ is an uncountable set which can be
either bounded or unbounded, and the random vector $\bw[t]$ is defined in a
probability space $(\mathcal{W},\sigma( \mathcal{W}), \Prob_{\bw})$.  The random
vector $\bw[t]$ is assumed to be absolutely continuous with a known density function
$\psi_{\bw}$.  The disturbance process $\bw[\cdot]$ is assumed to be a random
process with an independent and identical distribution (IID). 

The dynamics in \eqref{eq:sys_orig} are quite general and includes affine noise
perturbed LTI discrete-time systems with known state-feedback based inputs.  An
additional affine term in \eqref{eq:sys_orig} can include affine noise perturbed
LTI discrete-time systems with known open-loop controllers.  For time
$\tau\in[1,T]$,
\begin{align} 
    \bx[\tau]&=A^\tau\bar{x}_0+ \mathscr{C}_{n\times (\tau p)} \bW\label{eq:lin_traj}
\end{align} 
with $\mathscr{C}_{n\times (\tau p)}=[B\ AB\ A^2B\ \ldots\ A^{\tau-1}B]\in
\mathbb{R}^{n\times (\tau p)}$ and $\bW={[\bw^\top[\tau-1]\ \bw^\top[\tau-2]\
\ldots\ \bw^\top[0]]}^\top$ as a random vector defined by the sequence of random
vectors $\{\bw[t]\}_{t=0}^{t=\tau-1}$.   For any given $\tau$, the random vector
$\bW$ is defined in the product space $(\mathcal{W}^\tau,
\sigma(\mathcal{W}^\tau), \Prob_{\bW})$ where $\mathcal{W}^\tau=
\bigtimes_{t=0}^\tau\mathcal{W}$ and $\Prob_{\bW}=\prod_{t=0}^\tau \Prob_{\bw}$,
$\psi_{\bW}=\prod_{t=0}^\tau \psi_{\bw}$ by the IID assumption of the random
process $\bw[\cdot]$.  From \eqref{eq:lin_traj}, the state $\bx[\cdot]$ is a
random process with the random vector at each instant $\bx[t]$ defined in the
probability space $(\mathcal{X}, \sigma(\mathcal{X}),
\Prob_{\bx}^{t,\bar{x}_0})$ where the probability measure
$\Prob_{\bx}^{t,\bar{x}_0}$ is induced from $\Prob_{\bW}$.  We denote the random
process originating from $\bar{x}_0$ as $\bxi[\cdot;\bar{x}_0]$ where for all
$t$, $\bxi[t;\bar{x}_0]=\bx[t]$, and let $\bar{Z}={[\bar{z}^\top[\tau-1]\
        \bar{z}^\top[\tau-2]\ \ldots\ \bar{z}^\top[0]]}^\top\in \mathbb{R}^{\tau
        p}$ denote a realization of the random vector $\bW$. 

An iterative method for the forward stochastic reachability analysis (FSR
analysis) is given in \cite{Baisravan2017ACC}~\cite[Section
10.5]{lasota2013chaos}.  However, for systems perturbed by continuous random
variables, the numerical implementation of the iterative approach becomes
erroneous for larger time instants due to the iterative numerical evaluation of
improper integrals, motivating the need for an alternative implementable
approach.

\begin{prob}
    Given the dynamics \eqref{eq:sys_orig} with initial state $\bar{x}_0$, 
   construct analytical expressions at time instant $\tau$ for
   \begin{enumerate}
   \item the smallest closed set that covers all the reachable states
       (i.e., the forward stochastic reach set), and
   \item the probability measure over the forward stochastic reach set (i.e.,
       the forward stochastic reach probability measure) 
   \end{enumerate}
that do not require an iterative approach.\label{prob:basic} 
\end{prob}

We are additionally interested in applying the forward stochastic reachable set
(FSR set) and probability measure (FSRPM) to the problem of capturing a
non-adversarial target.  Specifically, we seek a convex formulation to the
problem of capturing a non-adversarial target. This requires convexity of the
FSR set and concavity of the objective function defined on the probability of
successful capture. 

\begin{prob}\label{prob:stochtarget}
For a finite time horizon, find a) a convex formulation for the maximization of the probability of capture of
a non-adversarial target with known
stochastic dynamics and initial state, and b) the resulting optimal controller that a deterministic robot must employ when there is a non-zero probability of capture.

\end{prob}

\begin{customthm}{2.a}
   Characterize the sufficient conditions for log-concavity of the FSRPM and
   convexity of the FSR set.\label{prob:suff_cond}
\end{customthm}



\section{Forward stochastic reachability analysis}
\label{sec:computeReach}

The existence of forward stochastic reach probability density (FSRPD) for systems of the form \eqref{eq:sys_orig} has
been demonstrated in~\cite[Section 10.5]{lasota2013chaos}.
For any $\tau\in[1,T]$, the probability of the state reaching a set $ \mathcal{S}\in\sigma(
\mathcal{X})$ at time $\tau$ starting at $\bar{x}_0$ is
defined using the FSRPM $\Prob_{\bx}^{\tau,\bar{x}_0}$, 
\begin{align}
    \Prob_{\bx}^{\tau,\bar{x}_0}\{\bx[\tau]\in
    \mathcal{S}\}&=\int_\mathcal{S}\psi_{\bx}(\bar{y};\tau,\bar{x}_0)d\bar{y},\quad\bar{y}\in
        \mathbb{R}^n.\label{eq:psi_def}
\end{align}
Since the disturbance set $ \mathcal{W}$ is uncountable, we focus on the
computation of the FSRPD $\psi_{\bx}$, and use \eqref{eq:psi_def} to link it to
the FSRPM.  We have discussed the countable case in~\cite{Baisravan2017ACC}.

We define the forward stochastic reach set (FSR set) as the support of the
random vector $\bxi[\tau;\bar{x}_0]=\bx[\tau]$ at $\tau\in[1,T]$ when the
initial condition is $\bar{x}_0\in \mathcal{X}$. From \eqref{eq:support_def3},
for a continuous FSRPD,
\begin{align}
    \FSRset(\tau, \bar{x}_0)&=\mathrm{cl}(\{\bar{y}\in
\mathcal{X}\vert\psi_{\bx}(\bar{y};\tau, \bar{x}_0)> 0\})\subseteq \mathcal{X}.\label{eq:FSRset}
\end{align}
\begin{lem}\label{lem:FSRmembership}
    $\mathcal{S}\cap \FSRset(\tau, \bar{x}_0)= \emptyset  \Rightarrow 
    \Prob_{\bx}^{\tau,\bar{x}_0}\{\bx[\tau]\in
    \mathcal{S}\} =0.$
\end{lem}
\begin{proof}
Follows from \eqref{eq:support_def2}.
\end{proof}
Note that when the disturbance set $\mathcal{W}$ is unbounded, the definition of
the FSR set \eqref{eq:FSRset} might trivially become $\mathbb{R}^n$.
Also, for uncountable $ \mathcal{W}$, the probability of the state taking a
particular value is zero, and therefore, the superlevel sets of the FSRPD do not
have the same interpretation as in the countable case~\cite{Baisravan2017ACC}.
However, given the FSRPD, we can obtain the likelihood that the state of
\eqref{eq:sys_orig} will reach a particular set of interest via
\eqref{eq:psi_def} and the FSR set via \eqref{eq:FSRset}.

\subsection{Iterative method for reachability analysis}
\label{sub:iter}

We extend the iterative approach for the FSR analysis proposed
in~\cite{Baisravan2017ACC} for a nonlinear discrete-time systems with discrete
random variables to a linear discrete-time system with continuous random
variables.  This discussion, inspired in part by~\cite[Section
10.5]{lasota2013chaos}, helps to develop proofs presented later.

Assume that the system matrix $A$ of \eqref{eq:sys_orig} is invertible. This
assumption holds for continuous-time systems which have been discretized via
Euler method. For $\tau\in[0,T-1]$, we have from \eqref{eq:sys_orig} and
Property P1, 
\begin{align}
    \psi_{\bx}(\bar{y};\tau+1,\bar{x}_0)&=\big(\psi_{A\bx}(\cdot;\tau,\bar{x}_0)\ast\psi_{B\bw}(\cdot)\big)(\bar{y})\label{eq:recurs_FSRPD}
\end{align}
with
\begin{align}
    \psi_{A\bx}(\bar{y};\tau,\bar{x}_0)&={\vert A
\vert}^{-1}\psi_{\bx}(A^{-1}\bar{y};\tau,\bar{x}_0)\label{eq:probTransform}
\end{align}
$\psi_{A\bx}(\bar{y};\tau,\bar{x}_0)={(\det{A})}^{-1}\psi_{\bx}(A^{-1}\bar{y};\tau,\bar{x}_0)$
from~\cite[Example 8.9]{gubner_probability_2006} for $\tau\geq 1$,
$\psi_{A\bx}(\bar{y};0,\bar{x}_0)=\delta(\bar{y}-A\bar{x}_0)$ where
$\delta(\cdot)$ is the Dirac-delta function~\cite[Chapter
5]{bracewell_fourier_1986}, and $\psi_{B\bw}$ as the probability density of the
random vector $B\bw$. We use Property P2 and \eqref{eq:cfun_ift} to obtain
$\psi_{B\bw}$~\cite[Corollary 1]{stein1971introduction}.
Equation \eqref{eq:recurs_FSRPD} is a special case of the result
in~\cite[Section 10.5]{lasota2013chaos}.
We extend the FSR set computation presented in~\cite{Baisravan2017ACC} in the
following lemma. 
\begin{lem}\label{lem:FSRset_oplus} 
    For $\tau\in[1,T]$, closed disturbance set $ \mathcal{W}$, and the system in
    \eqref{eq:sys_orig} with initial condition $\bar{x}_0$, $\FSRset(\tau,
    \bar{x}_0)\subseteq A(\FSRset(\tau-1, \bar{x}_0)) \oplus B
    \mathcal{W}=\{A^\tau\bar{x}_0\} \oplus \mathscr{C}_{n\times (\tau p)}
    \mathcal{W}^\tau$.
\end{lem}
\begin{proof}
    Follows from \eqref{eq:sys_orig}, \eqref{eq:lin_traj}, and Property P1.
\end{proof}
Lemma~\ref{lem:FSRset_oplus} allows the use of existing reachability analysis
schemes designed for bounded non-stochastic disturbance
models~\cite{kvasnica2015reachability, ellipsoid, girard2005reachability} for
overapproximating FSR sets.
Also, \eqref{eq:FSRset} and \eqref{eq:recurs_FSRPD} provide an
iterative method for exact FSR analysis.

Note that \eqref{eq:recurs_FSRPD} is an improper integral which must be
solved iteratively.  For densities whose convolution integrals are difficult to
obtain analytically, we would need to rely on numerical integration (quadrature) techniques.
Numerical evaluation of multi-dimensional improper integrals is computationally
expensive~\cite[Section 4.8]{press2007numerical}. 
Moreover, the quadratures in this
method will become increasingly erroneous for larger values of $\tau\in[1,T]$ due
to the iterative definition.  These disadvantages motivate
the need to solve Problem~\ref{prob:basic} --- an approach that provides
analytical expressions of the FSRPD, and thereby reduce the number of
quadratures required. 
The iterative method performs well with discrete random
vectors as in \cite{Baisravan2017ACC} because discretization for computation
can be exact, however, this is clearly not true when the disturbance set is
uncountable.  


\subsection{Efficient reachability analysis via characteristic functions}
\label{sub:FT}

We employ Fourier transformation to provide analytical expressions of the FSRPD
at any instant $\tau\in[1,T]$. This method involves computing a single
integral for the time instant of interest $\tau$ as opposed to the iterative approach in Subsection~\ref{sub:iter}. We also show that for certain
disturbance distributions like the Gaussian distribution, an explicit expression for
the FSRPD can be obtained.

By Property~P3 and the IID assumption on the random process $\bw[\cdot]$, the
CF of the random vector $\bW$ is 
\begin{align}
    \Psi_{\bW}(\bar{\alpha})&=\prod_{t=0}^{t=\tau-1}\Psi_{\bw}(\bar{\alpha}_t)\label{eq:cfun_W}
\end{align}
where $\bar{\alpha}={[\bar{\alpha}_0^\top\ \bar{\alpha}_1^\top\ \ldots\
\bar{\alpha}_{\tau-1}^\top]}^\top\in \mathbb{R}^{(\tau p)},\ \bar{\alpha}_t\in
\mathbb{R}^{p}$ for all $\tau\in[0,\tau-1]$. As seen in \eqref{eq:lin_traj}, the
random vector $\bW$ concatenates the disturbance random process $\bw[t]$ over
$t\in[0,\tau-1]$. 

\begin{thm}
    For any time instant $\tau\in[1,T]$ and an initial state $\bar{x}_0\in
    \mathcal{X}$, the FSRPD $\psi_{\bx}(\cdot;\tau,\bar{x}_0)$ of
    \eqref{eq:sys_orig} is given by 
    \begin{align}
        \Psi_{\bx}(\bar{\alpha};\tau,\bar{x}_0)&=\mathrm{exp}\left({j\bar{\alpha}^\top
(A^\tau\bar{x}_0)}\right)
    \Psi_{\bW}(\mathscr{C}_{n\times (\tau p)}^\top\bar{\alpha})\label{eq:cfun_xtau}
    \\
    \psi_{\bx}(\bar{y};\tau,\bar{x}_0)&=\mathscr{F}^{-1}\left\{\Psi_{\bx}(\bar{\alpha};\tau,\bar{x}_0)\right\}(-\bar{y})\label{eq:cfun_psitau}
\end{align}
where $\bar{y}\in \mathcal{X},\bar{\alpha}\in \mathbb{R}^{n\times
1}$.\label{thm:FSRPD_def}
\end{thm}
\begin{proof}
    Follows from Property P2,~\eqref{eq:cfun_ift}, and~\eqref{eq:lin_traj}.
\end{proof}

Theorem~\ref{thm:FSRPD_def} provides an analytical expression for the FSRPD.
Theorem~\ref{thm:FSRPD_def} holds even if we relax the identical distribution
assumption on the random process $\bw[t]$ to a time-varying independent
disturbance process, provided $\Psi_{\bw[t]}(\cdot)$ is known for all
$t\in[0,\tau-1]$. Using Property P2, Theorem~\ref{thm:FSRPD_def} can also be
easily extended to include affine noise perturbed LTI discrete-time systems with
known open-loop controllers. 

Note that the computation of the FSRPD via Theorem~\ref{thm:FSRPD_def} does not
require gridding of the state space, hence mitigating the curse of
dimensionality associated with the traditional gridding-based approaches.  When
the CF $\Psi_{\bx}(\bar{\alpha};\tau,\bar{x}_0)$ has the structure of known
Fourier transforms, Theorem~\ref{thm:FSRPD_def} can be used to provide explicit
expressions for the FSRPD (see Proposition~\ref{prop:FSRPD_gauss}).  In systems
where the inverse Fourier transform is not known, the evaluation of
\eqref{eq:cfun_psitau} can be done via any quadrature techniques that can handle
improper integrals. Alternatively, the improper integral can be approximated by
the quadrature of an appropriately defined proper integral~\cite[Chapter
4]{press2007numerical}. For high-dimensional
systems, performance is affected by the scalability of quadrature schemes with
dimension. However, Theorem~\ref{thm:FSRPD_def} still requires only a single
$n$-dimensional quadrature for any time instant of interest $\tau\in[1,T]$. On
the other hand, the iterative method proposed in Subsection~\ref{sub:iter}
requires $\tau$ quadratures, each $n$-dimensional, resulting in higher
computational costs and degradation in accuracy as $\tau$ increases.

One example of a CF with
known Fourier transforms arises in Gaussian distributions.
We use Theorem~\ref{thm:FSRPD_def} to derive an
explicit expression for the FSRPD of \eqref{eq:sys_orig} when perturbed by a
Gaussian random vector.  Note that the FSRPD in this
case can also be computed using the well-known properties on linear
combination of Gaussian random vectors~\cite[Section 9]{gubner_probability_2006}
or the theory of Kalman-Bucy filter~\cite{dorato1994linear}.   
\begin{prop}\label{prop:FSRPD_gauss}
    The system trajectory of \eqref{eq:sys_orig} with initial condition
    $\bar{x}_0$ and noise process $\bw \sim
    \mathcal{N}(\bar{\mu}_{\bw},\Sigma_{\bw})\in \mathbb{R}^p$ is 
    \begin{align}
        \bxi[\tau;\bar{x}_0]&\sim \mathcal{N}(\bar{\mu}[\tau],\Sigma[\tau])\label{eq:FSRPD_linear_gaussian}
    \end{align}
    where $\tau\in[1,T]$ and
    \begin{align}
        \bar{\mu}[\tau]&=A^\tau\bar{x}_0+ \mathscr{C}_{n\times (\tau
p)}(\bar{1}_{\tau\times 1}\otimes \bar{\mu}_{\bw}),\label{eq:mu_robotG_t}\\
\Sigma[\tau]&=\mathscr{C}_{n\times (\tau p)}( I_\tau\otimes\Sigma_{\bw})\mathscr{C}_{n\times
(\tau p)}^\top\label{eq:sigma_robotG_t}.
    \end{align}
\end{prop}
\begin{proof}
For $\bar{\alpha}\in \mathbb{R}^{p}$, the CF of a multivariate Gaussian random vector $\bw$ is~\cite[Section 9.3]{gubner_probability_2006}
\begin{align}
    \Psi_{\bw}(\bar{\alpha})&=\mathrm{exp}\left({j\bar{\alpha}^\top\bar{\mu}_{\bw}-\frac{\bar{\alpha}^\top\Sigma_{\bw}\bar{\alpha}}{2}}\right).
    \label{eq:cfun_gauss} 
\end{align}
From the IID assumption of $\bw[\cdot]$, Property P3, and \eqref{eq:cfun_gauss}, the
CF of $\bW$ is
\begin{align}
    \Psi_{\bW}(\bar{\alpha})&=\prod_{t=0}^{t=\tau-1}\mathrm{exp}\left({j\bar{\alpha}_t^\top\bar{\mu}_{\bw}-\frac{\bar{\alpha}_t^\top\Sigma_{\bw}\bar{\alpha}_t}{2}}\right)
    \nonumber \\
    &=\mathrm{exp}\left({j\bar{\alpha}^\top( \bar{1}_{\tau\times 1}\otimes \bar{\mu}_{\bw})-\frac{\bar{\alpha}^\top(I_\tau\otimes\Sigma_{\bw} )\bar{\alpha}}{2}}\right)
    \nonumber
\end{align}
where
$\bar{\alpha}=(\bar{\alpha}_0,\bar{\alpha}_1,\ldots,\bar{\alpha}_{\tau-1})\in
\mathbb{R}^{(\tau p)}$ with $\bar{\alpha}_t\in \mathbb{R}^{p}$. Here,
$\bar{1}_{p\times q}\in \mathbb{R}^{p\times q}$ is a matrix with all entries as
$1$, and $I_n$ is the identity matrix of dimension $n$.
By~\eqref{eq:cfun_psitau} and \eqref{eq:cfun_gauss}, $\bW\sim
\mathcal{N}(\bar{1}_{\tau\times 1}\otimes
\bar{\mu}_{\bw},I_\tau\otimes\Sigma_{\bw})$~\cite[Corollary 1.22]{stein1971introduction}. From \eqref{eq:cfun_xtau}, we see
that for $\bar{\beta}\in \mathbb{R}^{n}$, 
\begin{align}
    \Psi_{\bx}(\bar{\beta};\tau,\bar{x}_0)&=\mathrm{exp}\left({j\bar{\beta}^\top
    (A^\tau \bar{x}_0)}\right) \Psi_{\bW}(\mathscr{C}_{n\times (\tau p)}^\top\bar{\beta})) \nonumber \\
    &=\mathrm{exp}\left(j\bar{\beta}^\top(A^\tau \bar{x}_0+ \mathscr{C}_{n\times
    (\tau p)}(\bar{1}_{\tau\times 1}\otimes\bar{\mu}_{\bw})) \right)\times \nonumber
    \\
    \mathrm{exp}&\left(-\frac{\bar{\beta}^\top\mathscr{C}_{n\times (\tau
p)}(I_\tau\otimes\Sigma_{\bw})\mathscr{C}_{n\times (\tau
p)}^\top\bar{\beta}}{2}\right). \label{eq:cfun_gauss_robotG}
\end{align}
Equation~\eqref{eq:cfun_gauss_robotG} is the CF of a
multivariate Gaussian random vector~\cite[Corollary 1.22]{stein1971introduction}, and we obtain  $\mu_G[\tau]$
and $\Sigma_G[\tau]$ using \eqref{eq:cfun_gauss}.
\end{proof}

Depending on the system dynamics and time instant of interest $\tau$, we can
have $\mathrm{rank}( \mathscr{C}_{n\times (\tau p)})<n$. In such cases, the
support of the random vector $\bx[\tau]$ will be restricted to sets of lower
dimension than $n$~\cite[Section 8.5]{chow1997probability}, and certain marginal
densities can be Dirac-delta functions. For example, we see that turning off the
effect of disturbance in \eqref{eq:sys_orig} (setting $B=0 \Rightarrow
\mathscr{C}_{n\times (\tau p)}=0$ in Theorem~\ref{thm:FSRPD_def}) yields
$\psi_{\bx}(\bar{y};\tau,\bar{x}_0)=\delta(\bar{y}-A^\tau\bar{x}_0)$, the
trajectory of the corresponding deterministic system. We have used the relation
$\mathscr{F}\big\{\delta(\bar{y}-\bar{y}_0)\big\}(\bar{\alpha})$ $=
\exp{(j\bar{\alpha}^\top\bar{y_0})}$~\cite[Chapters 5,6]{bracewell_fourier_1986}.

Theorem~\ref{thm:FSRPD_def} and \eqref{eq:FSRset} provide an analytical
expression for the FSRPD and the FSR set respectively, and thereby solve
Problem~\ref{prob:basic} for any density function describing the stochastics of
the perturbation $\bw[t]$ in \eqref{eq:sys_orig}. 

\subsection{Convexity results for reachability analysis}
\label{sub:cvx_FSRPD}

For computational tractability, it is useful to study the convexity properties
of the FSRPD and the FSR sets.  We define the random vector $\bw_B=B\bw$ with
density $\psi_{\bw_B}$. 

\begin{lem}{\cite[Lemma 2.1]{dharmadhikari1988unimodality}}
    If $\psi_{\bw}$ is a log-concave distribution, then $\psi_{\bw_B}$ is a
    log-concave distribution.\label{lem:logconcave_v}
\end{lem}

\begin{thm}\label{thm:FSRPD_log_concave}
    If $\psi_{\bw}$ is a log-concave distribution, and $A$ in
    \eqref{eq:sys_orig} is invertible, then the FSRPD
    $\psi_{\bx}(\bar{y};\tau,\bar{x}_0)$ of \eqref{eq:sys_orig} is log-concave
    in $\bar{y}$ for every $\tau\in[1,T]$.
\end{thm}
\begin{proof}
    We prove this theorem via induction.  First, we need to show that the base
    case is true, i.e, we need to show that $\psi_{\bx}(\bar{y};1,\bar{x}_0)$ is
    log-concave in $\bar{y}$. From \eqref{eq:sys_orig} and
    Lemma~\ref{lem:logconcave_v}, we have a log-concave density
    $\psi_{\bx}(\bar{y};1,\bar{x}_0)=\psi_{\bw_B}(\bar{y}-A\bar{x}_0)$ since
    affine transformations preserve log-concavity~\cite[Section 3.2.4]{boyd_convex_2004}.
    Assume for induction, $\psi_{\bx}(\bar{y};\tau,\bar{x}_0)$ is log-concave in
    $\bar{y}$ for some $\tau\in[1,T]$. We have
    log-concave $\psi_{A\bx}(\bar{y};\tau,\bar{x}_0)$ from~\eqref{eq:probTransform}.
    Since convolution preserves log-concavity~\cite[Section
    3.5.2]{boyd_convex_2004}, Lemma~\ref{lem:logconcave_v}
    and~\eqref{eq:recurs_FSRPD} complete the proof.
\end{proof}

\begin{corr}
    If $\psi_{\bw}$ is a log-concave distribution,  $\FSRset(\tau, \bar{x}_0)$
    of the system \eqref{eq:sys_orig} is convex for every $\tau\in [1,T],\
    \bar{x}_0\in \mathcal{X}$.\label{thm:FSRset_convex}
\end{corr}
\begin{proof}
    Follows from \eqref{eq:FSRset} and~\cite[Theorem 2.5]{dharmadhikari1988unimodality}.
\end{proof}
Theorem~\ref{thm:FSRPD_log_concave} and Corollary~\ref{thm:FSRset_convex} solve
Problem~\ref{prob:suff_cond}. 

\section{Reaching a non-adversarial target with stochastic dynamics}
\label{sec:stochTarget}

In this section, we will leverage the theory developed in this paper to solve
Problem~\ref{prob:stochtarget} efficiently.

We consider the problem of a controlled robot (R) having to capture a stochastically
moving non-adversarial target, denoted here by a goal robot (G). The robot R has controllable linear
dynamics while the robot G has uncontrollable linear dynamics, perturbed by an
absolutely continuous random vector. The robot R is said to capture robot G
if the robot G is inside a pre-determined set defined around the current position of
robot R.  We seek an \emph{open-loop} controller
(independent of the current state of robot G) for the
robot R which maximizes the probability of capturing robot G within the time
horizon $T$.  The information available to solve this problem are the position of
the robots R and G at $t=0$, the deterministic dynamics of the robot R, the
perturbed dynamics of the robot G, and the density of the perturbation.  We
consider a $2$-D environment, but our approach can be easily extended to higher
dimensions. We perform the FSR analysis in the
inertial coordinate frame. 

We model the robot R as a point mass system discretized in time,
\begin{align}
    \bar{x}_R[t+1]&=\bar{x}_R[t]+B_R\bar{u}_R[t]\label{eq:robotR_dyn}
\end{align}
with state (position) $\bar{x}_R[t]\in \mathbb{R}^2$, input $\bar{u}_R[t]\in
\mathcal{U}\subseteq \mathbb{R}^2$, input matrix $B_R=T_sI_2$ and
sampling time $T_s$. We define an open-loop control policy
$\bar{u}_R[t]=\pi_{\mathrm{open},\bar{x}_R[0]}[t]$ where
$\pi_{\mathrm{open},\bar{x}_R[0]}[t]$ depends on the initial condition,
that is,\\ $\pi_{\mathrm{open},\bar{x}_R[0]}:[0,T-1] \rightarrow \mathcal{U}$ is
a sequence of control actions for a given initial condition $\bar{x}_R[0]$.  Let
$ \mathcal{M}$ denote the set of all feasible control policies
$\pi_{\mathrm{open},\bar{x}_R[0]}$. From \eqref{eq:lin_traj},
\begin{align}
    \bar{x}_R[\tau+1]&=\bar{x}_R[0]+(\bar{1}_{1\times
\tau}\otimes B_R)\bar{\pi}_\tau,\ \tau\in[0,T-1] \label{eq:lin_traj_robotR}
\end{align}
with the input vector $\bar{\pi}_\tau={[\bar{u}_R^\top[\tau-1]\ \bar{u}_R^\top[\tau-2]\
\ldots\ \bar{u}_R^\top[0]]}^\top$, $\bar{\pi}_\tau \in \overline{
\mathcal{M}}_\tau\subseteq\mathbb{R}^{(2\tau)},$ and
$\bar{u}_R[t]=\pi_{\mathrm{open},\bar{x}_R[0]}[t]$.   

We consider two cases for the dynamics of the robot G: 1) point mass dynamics,
and 2) double integrator dynamics, both discretized in time and perturbed
by an absolutely continuous random vector. In the former case, we presume that
the velocity is drawn from a bivariate Gaussian distribution, 
\begin{subequations}
    \begin{align}
        \bx_G[t+1]&=\bx_G[t]+B_\mathrm{G,PM}\bv_G[t] \label{eq:robotG_dyn_PM}\\
        \bv_G[t] &\sim \mathcal{N}(\bar{\mu}_G^{\bv},\Sigma_G). \label{eq:disturb_PM}
    \end{align}\label{eq:robotG_PM}%
\end{subequations}
The state (position) is the random vector $\bx_G[t]$ in the probability space
$( \mathcal{X}, \sigma( \mathcal{X}), \Prob_{\bx_G}^{t,\bar{x}_G[0]})$ with $\mathcal{X}=
\mathbb{R}^2$, disturbance matrix $B_\mathrm{G,PM}=B_R$, and $\bar{x}_G[0]$ as
the known initial state of the robot G. The stochastic
velocity $\bv_G[t]\in \mathbb{R}^{2}$ has mean vector
$\bar{\mu}_G^{\bv}$, covariance matrix $\Sigma_G$ and the CF with $\bar{\alpha}\in \mathbb{R}^{2}$ is given in
\eqref{eq:cfun_gauss}.
In the latter case, acceleration in each direction is an independent exponential random
variable,
\begin{subequations}
    \begin{align}
        \bx_G[t+1]&=A_\mathrm{G,DI}\bx_G[t]+B_\mathrm{G,DI}\ba[t] \label{eq:robotG_dyn_DI}\\
        {(\ba[t])}_\mathrm{x} &\sim
        \mathrm{Exp}(\lambda_\mathrm{ax}),\quad{(\ba[t])}_\mathrm{y} \sim \mathrm{Exp}(\lambda_\mathrm{ay}) \label{eq:acc2}\\
    A_\mathrm{G,DI}&= I_2\otimes \left[ {\begin{array}{cc} 1  &  T_s \\ 
                                                0  & 1 \\ 
                            \end{array} } \right],
                                                          B_\mathrm{G,DI}
                                                          = I_2\otimes \left[ {\begin{array}{c}
                                              \frac{T_s^2}{2} \\ 
                                                          T_s \\
                                              \end{array} } \right]. \nonumber
    \end{align}\label{eq:robotG_DI}%
\end{subequations}
The state (position and velocity) is the random vector $\bx_G[t]$ in the
probability space $(\mathcal{X}_\mathrm{DI}, \sigma( \mathcal{X}_\mathrm{DI}),
\Prob_{\bx_G}^{t,\bar{x}_G[0]})$ with $\mathcal{X}_\mathrm{DI}= \mathbb{R}^4$
and $\bar{x}_G[0]$ as the known initial state of the robot G. The stochastic
acceleration $\ba[t]={[{(\ba[t])}_\mathrm{x}\ {(\ba[t])}_\mathrm{y}]}^\top\in
\mathbb{R}^{2}_+=[0,\infty)\times[0,\infty)$ has the following probability
        density and CF ($\bar{z}={[z_1\ z_2]}^\top\in
        \mathbb{R}^2_+=[0,\infty)\times[0,\infty), \bar{\alpha}={[\alpha_1\
                \alpha_2]}^\top\in \mathbb{R}^2$),
\begin{align}
    \psi_{\ba}(\bar{z})&=\lambda_\mathrm{ax}\lambda_\mathrm{ay} \exp{(-\lambda_\mathrm{ax} z_1-\lambda_\mathrm{ay}
z_2)}\label{eq:exp_f}\\
\Psi_{\ba}(\bar{\alpha})&=\frac{\lambda_\mathrm{ax}\lambda_\mathrm{ay}}{(\lambda_\mathrm{ax}-j\alpha_1)(\lambda_\mathrm{ay}-j\alpha_2)}\label{eq:cfun_exp}.
\end{align}
The CF $\Psi_{\ba}(\bar{\alpha})$ is defined using Property
P3 and the CF of the exponential given
in~\cite[Section 26]{billingsley_probability_1995}.

Formally, the robot R captures robot G at time $\tau$ if $\bx_G[\tau]\in$\\$
\mathrm{CaptureSet}(\bar{x}_R[\tau])$. In other words, the capture region of the
robot R is the $\mathrm{CaptureSet}(\bar{y})\subseteq  \mathbb{R}^2$ when robot
R is at $\bar{y}\in \mathbb{R}^2$. The optimization problem to solve
Problem~\ref{prob:stochtarget} is
\begin{align}
 \begin{array}{rl}
     \mbox{ProbA}: &  \begin{array}{rcl}
     \operatorname*{\mathrm{maximize}}& &
     \mathrm{CapturePr}_{\bar{\pi}}(\tau,\bar{\pi}_\tau;\bar{x}_R[0],\bar{x}_G[0])\\
     \mbox{subject to}& & (\tau,\bar{\pi}_\tau)\in[1,T]\times
     \overline{\mathcal{M}}_\tau
   \end{array}
 \end{array} \nonumber
\end{align}
where the decision variables are the time of capture $\tau$ and the control
policy $\bar{\pi}$, and the objective function $
\mathrm{CapturePr}_{\bar{\pi}}(\cdot)$ gives the probability of robot R capturing
robot G. By \eqref{eq:lin_traj_robotR}, an initial state
$\bar{x}_R[0]$ and the control policy $\bar{\pi}_t$ determines a unique
$\bar{x}_R[\tau]$ for every $\tau$. Using this observation, we define the objective
function $\mathrm{CapturePr}_{\bar{\pi}}(\cdot)$ in \eqref{eq:intprob}.
We obtain $\psi_{\bx_G}$ in \eqref{eq:intprob} using our solution to
Problem~\ref{prob:basic}, Theorem~\ref{thm:FSRPD_def}. 
\begin{figure*}
\begin{align}
\mathrm{CapturePr}_{\bar{\pi}}(\tau,\bar{\pi}_\tau;\bar{x}_R[0],\bar{x}_G[0]
) =\mathrm{CapturePr}_{\bar{x}_R}(\tau,\bar{x}_R[\tau];\bar{x}_G[0])
&= \Prob_{\bx_G}^{\tau,\bar{x}_G[0]}\{\bx_G[\tau]\in
\mathrm{CaptureSet}(\bar{x}_R[\tau])\} \nonumber \\
&=\int_{\mathrm{CaptureSet}(\bar{x}_R[\tau])}
\psi_{\bx_G}(\bar{y};\tau,\bar{x}_G[0])d\bar{y}.\label{eq:intprob}
\end{align}\rule{\textwidth}{0.5pt}
\end{figure*}

Problem \probref{ProbA} is equivalent (see~\cite[Section 4.1.3]{boyd_convex_2004})  to
\begin{align}
 \begin{array}{rl}
     \mbox{ProbB}: &  \begin{array}{rcl}
     \operatorname*{\mathrm{maximize}}& & \mathrm{CapturePr}_{\bar{x}_R}(\tau,\bar{x}_R[\tau];\bar{x}_G[0])\\
     \mbox{subject to}& & \left\{\begin{array}{rl}
        \tau&\in[1,T] \\
        \bar{x}_R[\tau]&\in \mathrm{Reach}_R(\tau;\bar{x}_R[0]) \\
    \end{array}\right.
    \end{array}
 \end{array} \nonumber
\end{align}
where the decision variables are the time of capture $\tau$ and the position of
the robot R $\bar{x}_R[\tau]$ at time $\tau$. From \eqref{eq:lin_traj_robotR}, we
define the reach set for the robot R at time $\tau$ as
\begin{align}
    \mathrm{Reach}_R(\tau;\bar{x}_R[0])&=\big\{\bar{y}\in
    \mathcal{X}\vert\exists\bar{\pi}_\tau\in \overline{\mathcal{M}}_\tau\mbox{ s.t.
}\bar{x}_R[\tau]=\bar{y}\big\}.\nonumber 
\end{align}
Several deterministic reachability computation tools are available for the
computation of $\mathrm{Reach}_R(\tau;\bar{x}_R[0])$
, like MPT~\cite{MPT3} and ET~\cite{ellipsoid}. We will now formulate
Problem~\probref{ProbB} as a convex optimization problem based on the results
developed in Subsection~\ref{sub:cvx_FSRPD}.
\begin{lem}{\cite{kvasnica2015reachability}}
If the input space $ \mathcal{U}$ is convex, the forward reach set
$\mathrm{Reach}_R(\tau;\bar{x}_R[0])$ is convex. \label{lem:cvx_reachR}
\end{lem}
\begin{prop}
    If $\psi_{\bw}$ is a log-concave distributions and $\mathrm{CaptureSet}(\bar{y})$
    is convex for all $\bar{y}\in \mathcal{X}$, then\\
    $\mathrm{CapturePr}_{\bar{x}_R}(\tau,\bar{y};\bar{x}_G[0])$ is log-concave in
    $\bar{y}$ for all $\tau$.\label{prop:yield_logconcave}
\end{prop}
\begin{proof}
    From Theorem~\ref{thm:FSRPD_log_concave}, we know that
    $\psi_{\bx}(\bar{y};\tau,\bar{x}_R[0])$ is log-concave in $\bar{y}$ for every
    $\tau$. The proof follows from \eqref{eq:intprob} 
    since the integration of a log-concave function over a convex set is
    log-concave~\cite[Section 3.5.2]{boyd_convex_2004}.
\end{proof}
\begin{rem}
The densities $\psi_{\bv}$ and $\psi_{\ba}$ are log-concave since multivariate
Gaussian density and exponential distribution 
(gamma distribution with shape parameter $p=1$) 
are log-concave, and
log-concavity is preserved for products~\cite[Sections 1.4, 2.3]{dharmadhikari1988unimodality}\cite[Section 3.5.2]{boyd_convex_2004}.      
\end{rem}

For any $\tau\in[1,T]$, Proposition~\ref{prop:yield_logconcave} and
Lemma~\ref{lem:cvx_reachR} ensure
\begin{align}
  \begin{array}{rl}
    \mbox{ProbC}: &  \begin{array}{rcl}
      \mbox{minimize}& &
      -\log(\mathrm{CapturePr}_{\bar{x}_R}(\tau,\bar{x}_R[\tau];\bar{x}_G[0])) \\
      \mbox{subject to}& & \bar{x}_R[\tau]\in \mathrm{Reach}_R(\tau;\bar{x}_R[0]) \\
    \end{array}
  \end{array}\nonumber 
\end{align}
is convex with the decision variable
$\bar{x}_R[\tau]$. Problem \probref{ProbC} is an equivalent convex optimization
problem of the partial maximization with respect to $\bar{x}_R[\tau]$ of Problem
\probref{ProbB} since we have transformed the original objective function with a
monotone function to yield a convex objective and the constraint sets are
identical~\cite[Section 4.1.3]{boyd_convex_2004}.  

We solve Problem~\probref{ProbB} by solving Problem~\probref{ProbC} for each
time instant $\tau\in[1,T]$ to obtain $\bar{x}_R^\ast[\tau]$ and compute the
maximum of the resulting finite set to get
$(\tau^\ast,\bar{x}_R^\ast[\tau^\ast])$. Since Problem~\probref{ProbB} could be
non-convex, this approach ensures a global optimum is found.
Note that in order to prevent taking the logarithm of zero,
we add an additional constraint to Problem~\probref{ProbC} 
\begin{align}
    \mathrm{CapturePr}_{\bar{x}_R}(\tau,\cdot;\bar{x}_G[0])&\geq
    \epsilon. \label{eq:add_const}
\end{align}
The constraint \eqref{eq:add_const} does not affect its convexity ($\epsilon$ is
a small positive number) from Proposition~\ref{prop:yield_logconcave} and the
fact that log-concave functions are quasiconcave. Quasiconcave functions have
convex superlevel sets~\cite[Sections 3.4, 3.5]{boyd_convex_2004}.

Using the optimal solution of Problem \probref{ProbB}, we can compute the
open-loop controller to drive the robot R from $\bar{x}_R[0]$ to
$\bar{x}^\ast_R[\tau^\ast]$ by solving Problem \probref{ProbD}.  Defining
$\mathscr{C}_R=(\bar{1}_{1\times (\tau^\ast-1)}\otimes B_R)$ from
\eqref{eq:lin_traj_robotR},
\begin{align}
    \begin{array}{rl}
      \mbox{ProbD}:  &  \begin{array}{rcl}\mbox{minimize}& & J_\pi(\bar{\pi}_{\tau^\ast}) \\
                        \mbox{subject to}& & \left\{\begin{array}{rl}
                        \bar{\pi}_{\tau^\ast}&\in \overline{\mathcal{M}}_{\tau^\ast} \\
                        \mathscr{C}_R\bar{\pi}_{\tau^\ast} &=\bar{x}_R[\tau^\ast]-\bar{x}_R[0] \\
                      \end{array}\right.
                    \end{array}\nonumber 
    \end{array}
\end{align}
where the decision variable is $\bar{\pi}_{\tau^\ast}$. The objective function
$J_\pi(\bar{\pi}_{\tau^\ast})=0$ provides a feasible open-loop controller, and
$J_\pi(\bar{\pi}_{\tau^\ast})=\bar{\pi}^\top_{\tau^\ast} \bar{R}
\bar{\pi}_{\tau^\ast},\bar{R}\in \mathbb{R}^{(2\tau^\ast)\times(2\tau^\ast)}$
provides an open-loop controller policy that minimizes the control effort while
ensuring that maximum probability of robot R capturing robot G is achieved.
Solving the optimization problems \probref{ProbB} and \probref{ProbD} answers
Problem~\ref{prob:stochtarget}.

Our approach to solving Problem~\ref{prob:stochtarget} is based on our solution
to Problem~\ref{prob:basic}, the Fourier transform based FSR analysis, and
Problem~\ref{prob:suff_cond}, the convexity results of the FSRPD and the FSR
sets presented in this paper.  In contrast, the iterative approach for the FSR
analysis, presented in Subsection~\ref{sub:iter}, would yield erroneous
$\mathrm{CapturePr}_{\bar{x}_R}(\tau,\cdot;\bar{x}_G[0])$ for larger values of
$\tau$ due to the heavy reliance on quadrature techniques. Additionally, the
traditional approach of dynamic programming based
computations~\cite{summers_verification_2010} would be prohibitively costly for
the large FSR sets encountered in this problem due to unbounded disturbances.
The numerical implementation of this work is discussed in
Subsection~\ref{sub:numImp}.

\begin{figure}
    \raggedleft
    \newcommand{\figwidthSnap}{0.28\textwidth}
    \begin{subfigure}{\figwidthSnap}
        \includegraphics[width=1\linewidth]{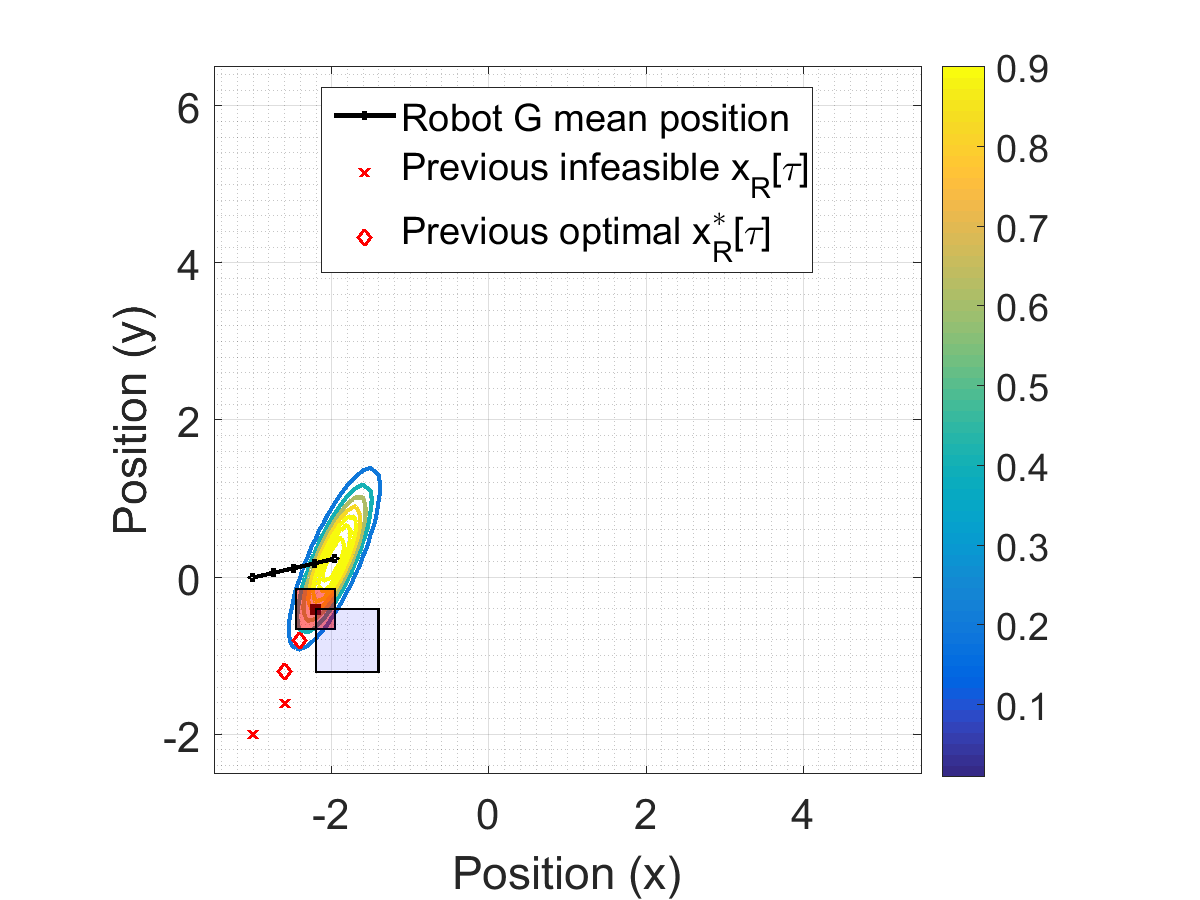}        
        \tikzset{
        every picture/.style={remember picture,baseline},
        }
        \begin{tikzpicture}[remember picture,overlay] 
            \node [text centered,text=black] at (-4.8em,7em) (A) {\footnotesize $\begin{aligned}
                \mathrm{Time}&=4\\
                \mathrm{CapturePr}_{\bar{x}_R}^\ast&=0.1571
            \end{aligned}$};
    \end{tikzpicture}
    \end{subfigure}

    \begin{subfigure}{\figwidthSnap}
        \includegraphics[width=1\linewidth]{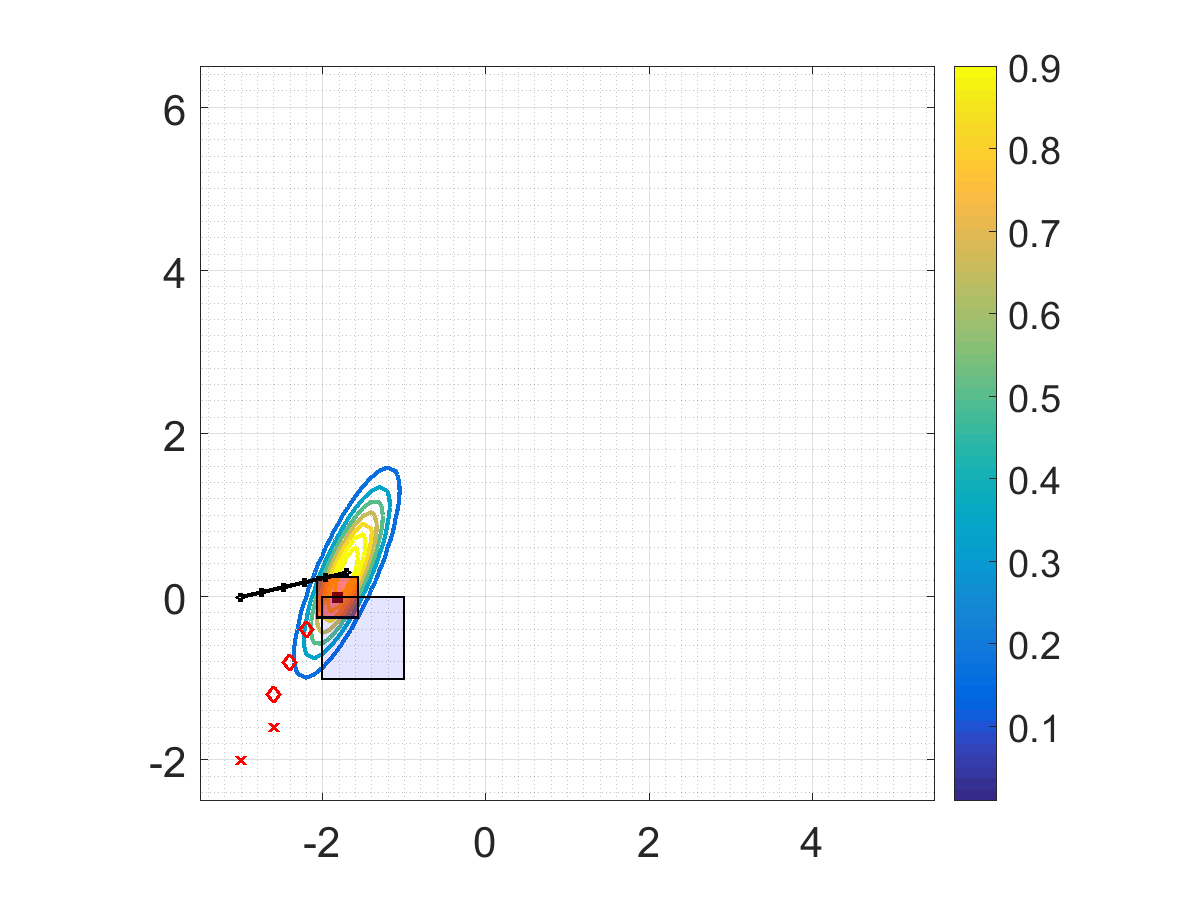}        
        \begin{tikzpicture}[remember picture,overlay] 
            \node [text centered,text=black] at (-4.8em,7em) (A) {\footnotesize $\begin{aligned}
                \mathrm{Time}&=5\\
                \mathrm{CapturePr}_{\bar{x}_R}^\ast&=0.219
            \end{aligned}$};
    \end{tikzpicture}
    \end{subfigure}
    ~ 
    \begin{subfigure}{\figwidthSnap}
        \includegraphics[width=1\linewidth]{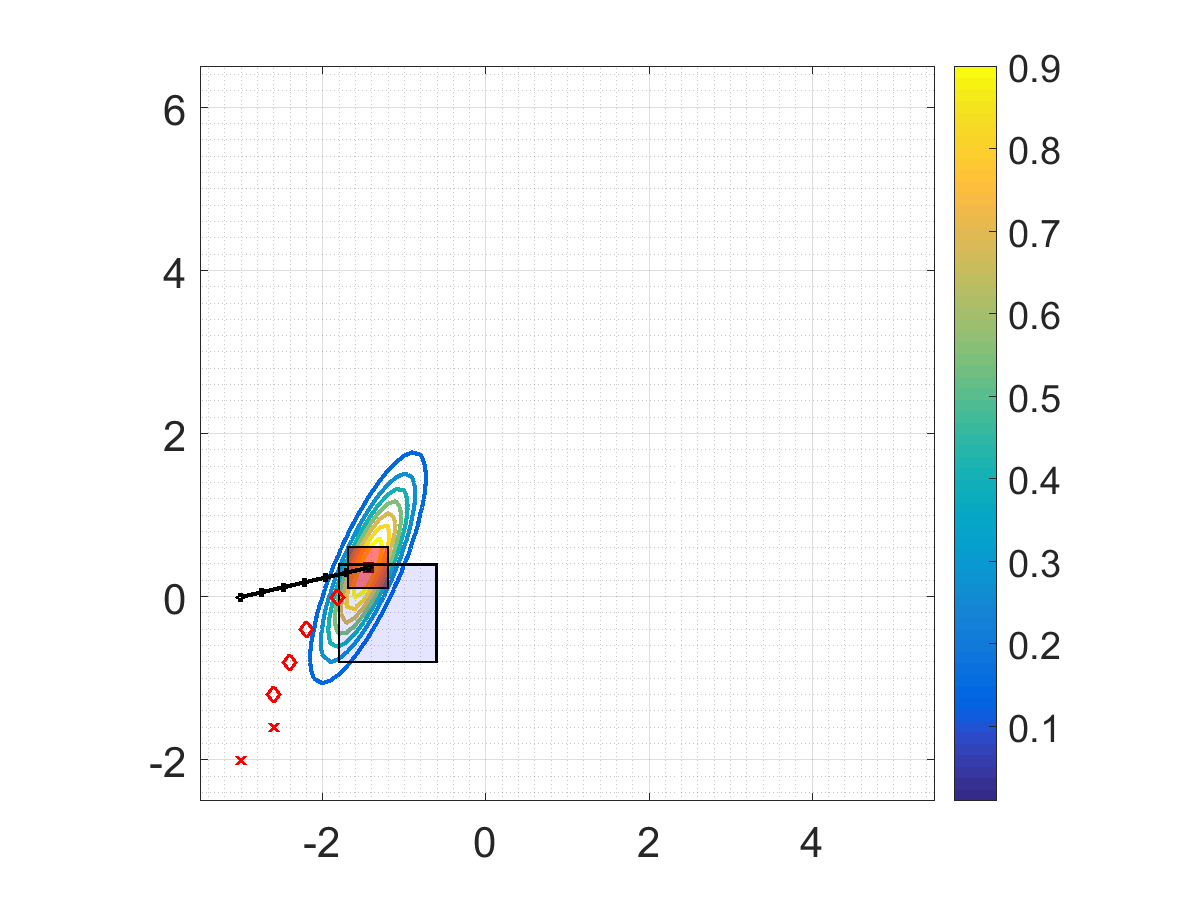} \label{fig:sim16}
        \begin{tikzpicture}[remember picture,overlay] 
            \node [text centered,text=black] at (-4.8em,7em) (A) {\footnotesize $\begin{aligned}
                \mathrm{Time}&=6\\
                \mathrm{CapturePr}_{\bar{x}_R}^\ast&=0.2124
            \end{aligned}$};
    \end{tikzpicture}
    \end{subfigure}    
    ~ 
    \begin{subfigure}{\figwidthSnap}
        \includegraphics[width=1\linewidth]{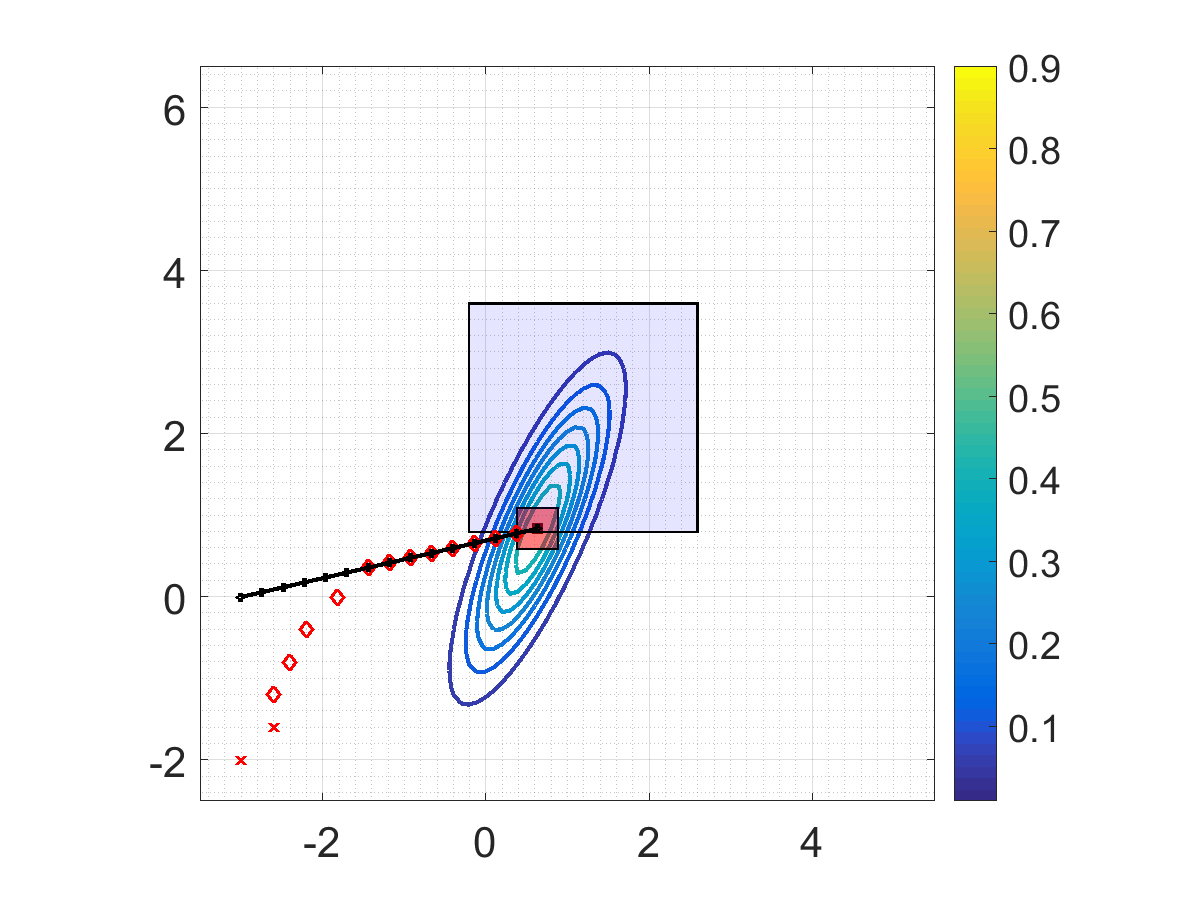}        
        \begin{tikzpicture}[remember picture,overlay] 
            \node [text centered,text=black] at (-4.8em,7em) (A) {\footnotesize $\begin{aligned}
                \mathrm{Time}&=14\\
                \mathrm{CapturePr}_{\bar{x}_R}^\ast&=0.1049
            \end{aligned}$};
    \end{tikzpicture}
    \end{subfigure}
    ~ 
    \begin{subfigure}{\figwidthSnap}
        \includegraphics[width=1\linewidth]{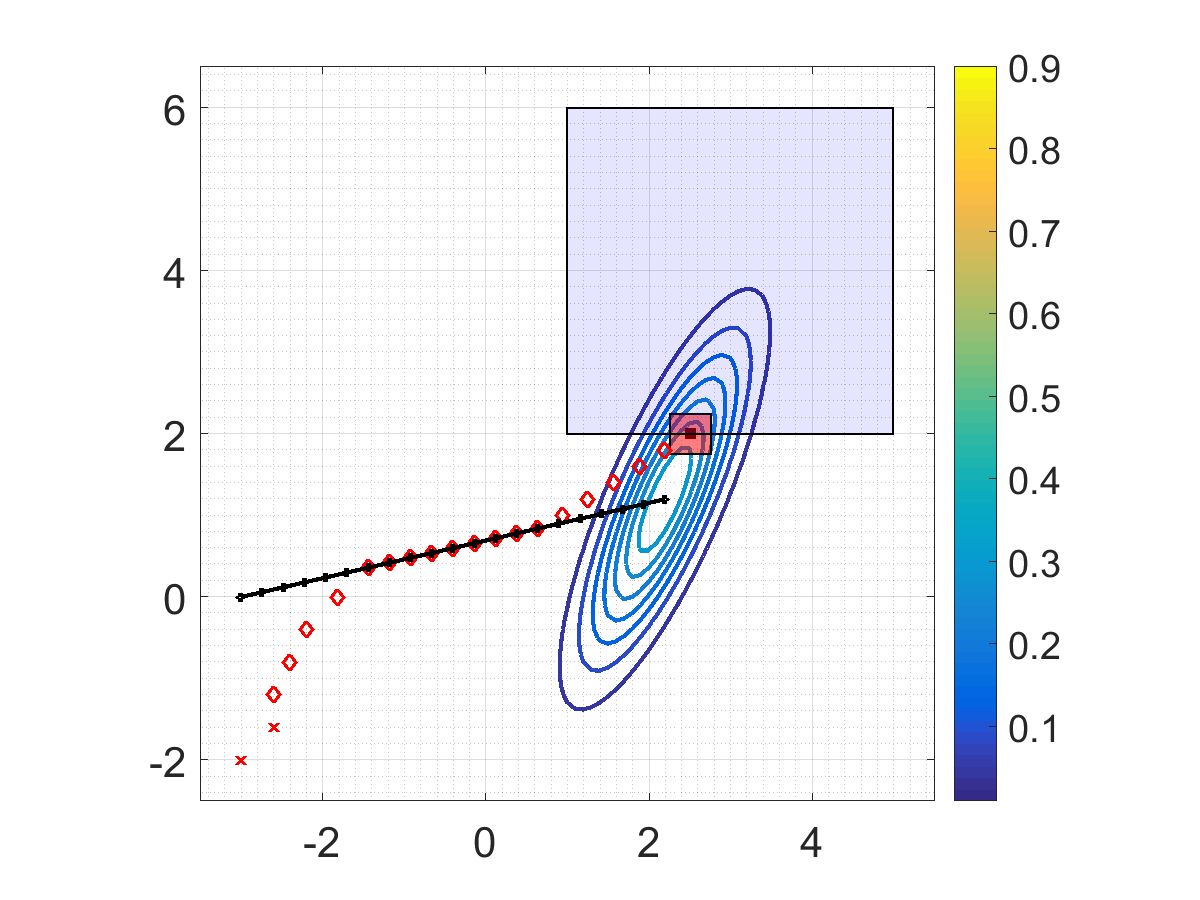} \label{fig:sim120}
        \begin{tikzpicture}[remember picture,overlay] 
            \node [text centered,text=black] at (-4.8em,7em) (A) {\footnotesize $\begin{aligned}
                \mathrm{Time}&=20\\
                \mathrm{CapturePr}_{\bar{x}_R}^\ast&=0.0624
            \end{aligned}$};
    \end{tikzpicture}
     \end{subfigure}    
     \caption{Snapshots of optimal capture positions of the robots G and R  when
         G has point mass dynamics~\eqref{eq:robotG_PM}. The blue line shows the
         mean position trajectory of robot G $\mu_G[\tau]$, the contour plot
         characterizes $\psi_{\bx_G}(\cdot;\tau,\bar{x}_G[0])$, the blue box shows
         the reach set of the robot R at time $\tau$
     $\mathrm{Reach}_R(\tau,\bar{x}_R[0])$, and the red box shows the capture
 region centered at $\bar{x}_R^\ast[\tau]$
 $\mathrm{CaptureSet}(\bar{x}_R^\ast[\tau])$.  }\label{fig:SIM2_gauss_snapshot}
\end{figure}
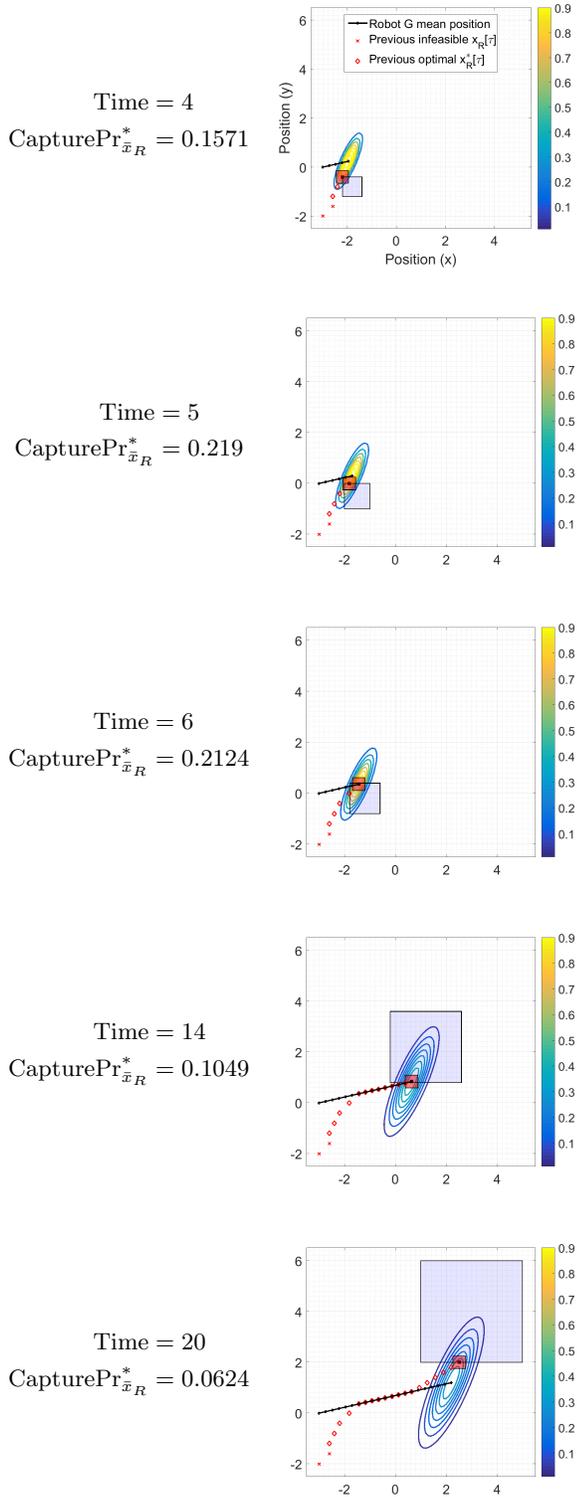

\subsection{Robot G with point mass dynamics}
\label{sub:gauss_PM}

\begin{figure}
    \centering 
\includegraphics[width=0.8\linewidth]{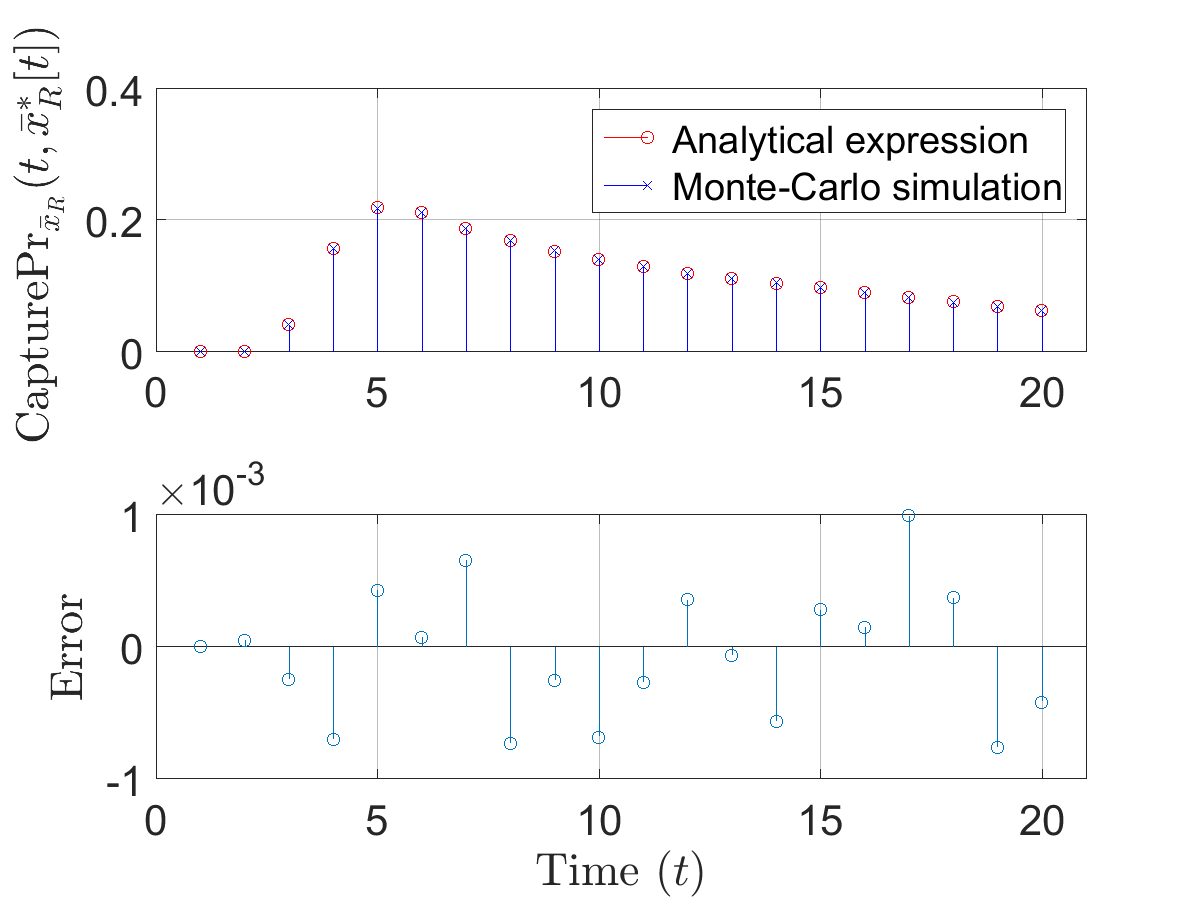} 
    \caption{Solution to Problem \probref{ProbC} for robot G dynamics in \eqref{eq:robotG_PM}, and validation of
        $\mathrm{CapturePr}_{\bar{x}_R}(\tau,\bar{x}_R^\ast[\tau];\bar{x}_G[0])$ via
        Monte-Carlo simulations. The optimal capture
        time is $\tau^\ast=5$ and the likelihood of capture is
$\mathrm{CapturePr}_{\bar{x}_R}(\tau^\ast,\bar{x}_R^\ast[\tau^\ast];\bar{x}_G[0])=0.219$.}\label{fig:SIM2_gauss_time}
\end{figure}

We solve Problem \probref{ProbB} for the system given by \eqref{eq:robotG_PM}.
Here, the disturbance set is $ \mathcal{W}= \mathbb{R}^2$.
\begin{lem}\label{lem:FSRset_gauss}
    For the system given in \eqref{eq:robotG_PM} and initial state of the robot G
    as $\bar{x}_G[0] \in \mathbb{R}^2$, $\FSRset_G(\tau,\bar{x}_G[0])=  \mathbb{R}^2$
    for every $\tau\in[1,T]$.
\end{lem}
\begin{proof}
    Follows from Proposition~\ref{prop:FSRPD_gauss} and \eqref{eq:FSRset}. 
\end{proof}
Proposition~\ref{prop:FSRPD_gauss} provides the FSRPD and
Lemma~\ref{lem:FSRset_gauss} provides the FSR set for the system
\eqref{eq:robotG_PM}.
The probability of successful capture of the robot G can be computed
using \eqref{eq:intprob} since the FSRPD $\psi_{\bx_G}(\cdot;\tau,\bar{x}_G[0])$ is
available. 

We implement the problem with the following parameters: $T_s=0.2$, $T=20$,
$\bar{\mu}_G^{\bv}=[1.3\ 0.3]^\top$, $\Sigma_G= \left[ {\begin{array}{cc} 0.5  &
0.8 \\ 0.8  & 2  \end{array} } \right]$, $\bar{x}_G[0]= [-3\ 0]^\top$,
$\bar{x}_R[0]=[-3\ -2]^\top$ and $ \mathcal{U}=[1,2]^2$. The capture region of
the robot R is a box centered about the position of the robot $\bar{y}$ with
edge length $2a$ ($a=0.25$) and edges parallel to the axes ---
$\mathrm{CaptureSet}(\bar{y})=\mathrm{Box}(\bar{y},a)$, a convex set. 
We use $J_\pi(\bar{\pi})=0$ in Problem~\probref{ProbD}.

Figure~\ref{fig:SIM2_gauss_snapshot} shows the evolution of the mean position of
the robot G and the optimal capture position for the robot R at time instants
$4,5,8,14,$ and $20$. The contour plots of
$\psi_{\bx_G}(\cdot;\tau,\bar{x}_G[0])$ are rotated ellipses since $\Sigma_E$ is
not a diagonal matrix. From \eqref{eq:mu_robotG_t}, the mean position of the
robot G moves in a straight line $\mu_G[\tau]$, as it is the trajectory of
\eqref{eq:robotG_dyn_PM} when the input is always $\bar{\mu}_G^{\bv}$. The
optimal time of capture is $\tau^\ast=5$, the optimal capture position is
$\bar{x}_R^\ast[\tau^\ast]=[-1.8\ 0]^\top$, and the corresponding probability of
robot R capturing robot G is $0.219$.  Note that at this instant, the reach set
of the robot R does not cover the current mean position of the robot G,
$\bar{\mu}[\tau^\ast]=[-1.7\ 0.3]^\top$ (Figure~\ref{fig:SIM2_gauss_snapshot}b).
While the reach set covers the mean position of robot G at the next time instant
$t=6$, the uncertainty in \eqref{eq:robotG_PM} causes the probability of
successful capture to further reduce (Figure~\ref{fig:SIM2_gauss_snapshot}c).
Counterintuitively, attempting to reach the mean $\mu_G[\tau]$ is not always best.
Figure~\ref{fig:SIM2_gauss_time} shows the optimal capture probabilities
obtained when solving Problem \probref{ProbC} for the
dynamics~\eqref{eq:robotG_PM}.

\subsection{Robot G with double integrator dynamics}
\label{sub:exp_DI}

\begin{figure}
    \raggedleft
    \newcommand{\figwidthSnap}{0.27\textwidth}
    \begin{subfigure}{\figwidthSnap}
        \includegraphics[width=1\linewidth]{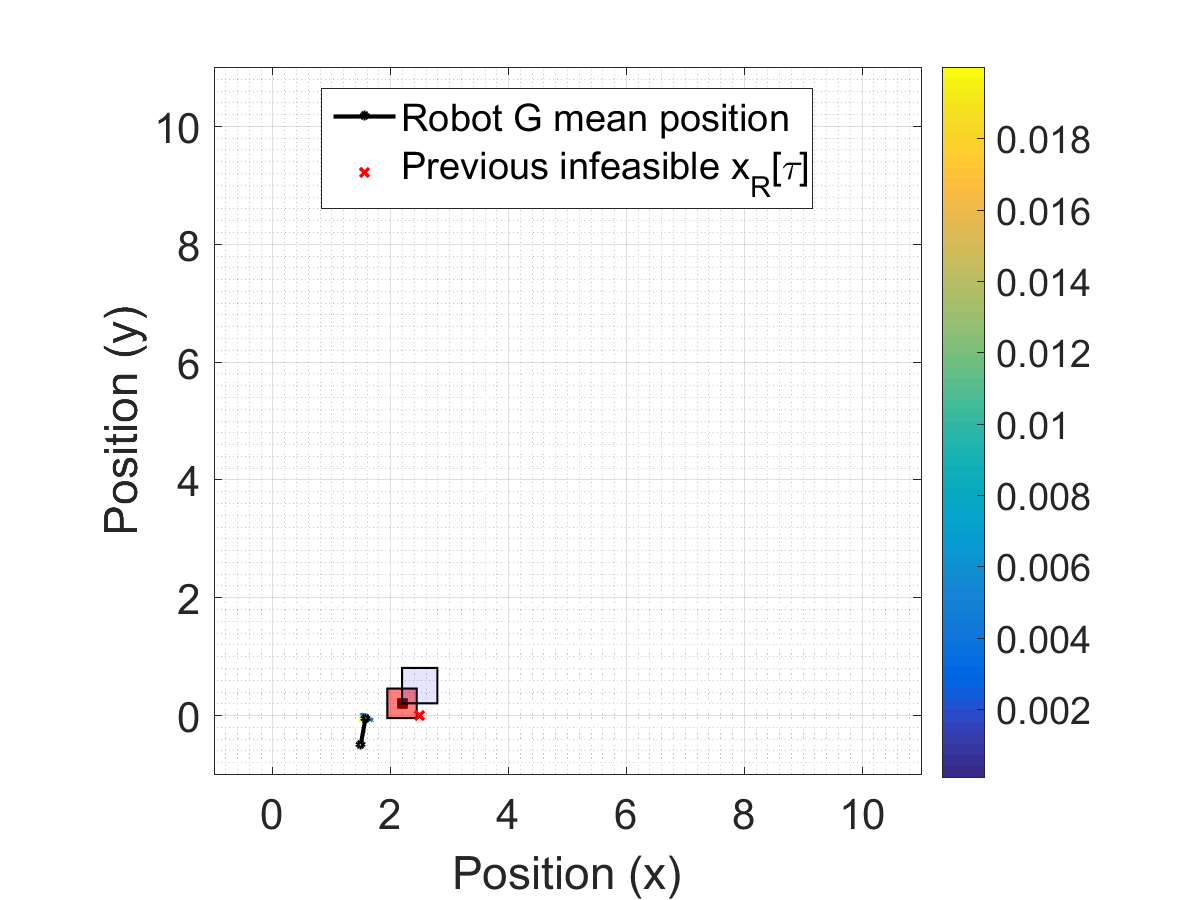}        
        \tikzset{
        every picture/.style={remember picture,baseline},
        }
        \begin{tikzpicture}[remember picture,overlay] 
            \node [text centered,text=black] at (-4.8em,7em) (A) {\footnotesize $\begin{aligned}
                \mathrm{Time}&=1\\
                    \mathrm{CapturePr}_{\bar{x}_R}^\ast&=0\\
                    (\mbox{Inf}&\mbox{easible})
            \end{aligned}$};
    \end{tikzpicture}
    \end{subfigure}

    \begin{subfigure}{\figwidthSnap}
        \includegraphics[width=1\linewidth]{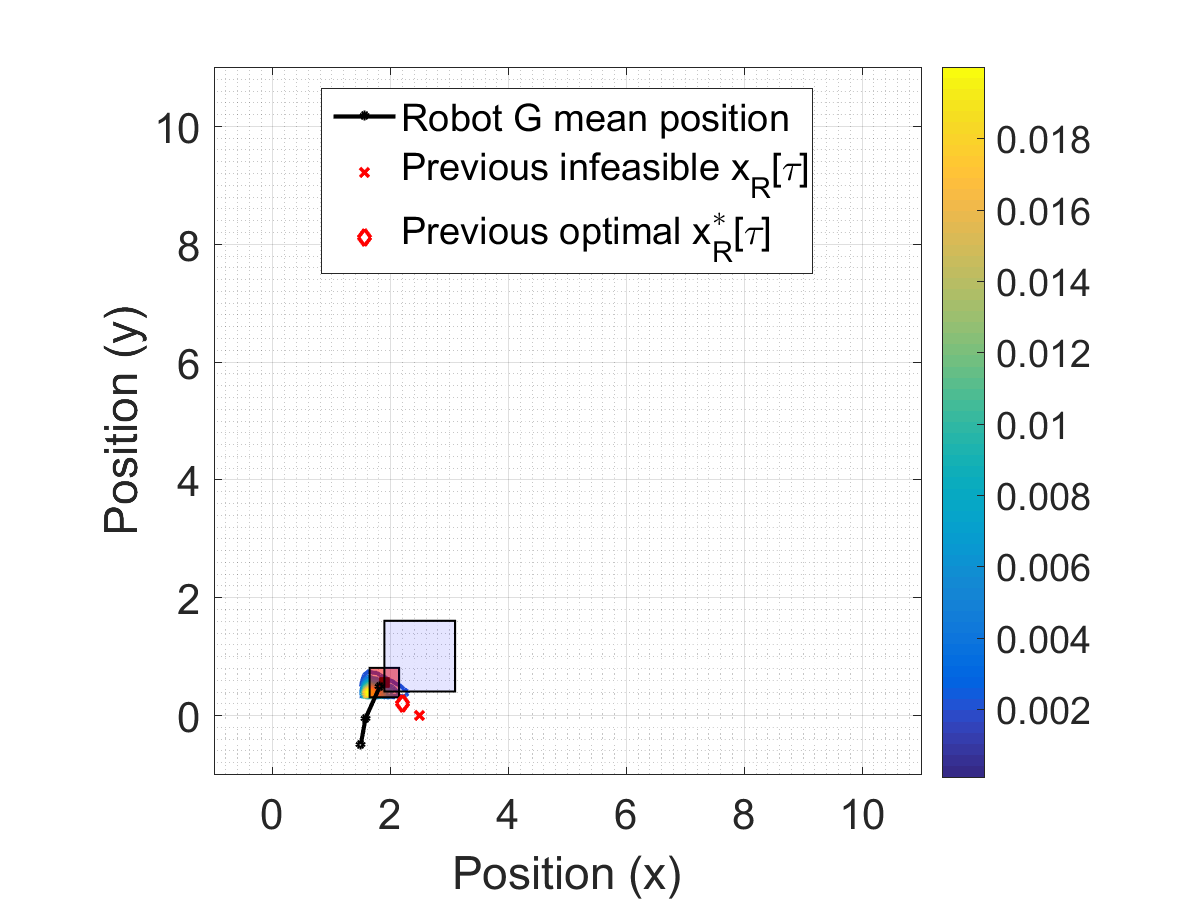}        
        \begin{tikzpicture}[remember picture,overlay] 
            \node [text centered,text=black] at (-4.8em,7em) (A) {\footnotesize $\begin{aligned}
                \mathrm{Time}&=2\\
                \mathrm{CapturePr}_{\bar{x}_R}^\ast&=0.6044
            \end{aligned}$};
    \end{tikzpicture}
    \end{subfigure}
    ~ 
    \begin{subfigure}{\figwidthSnap}
        \includegraphics[width=1\linewidth]{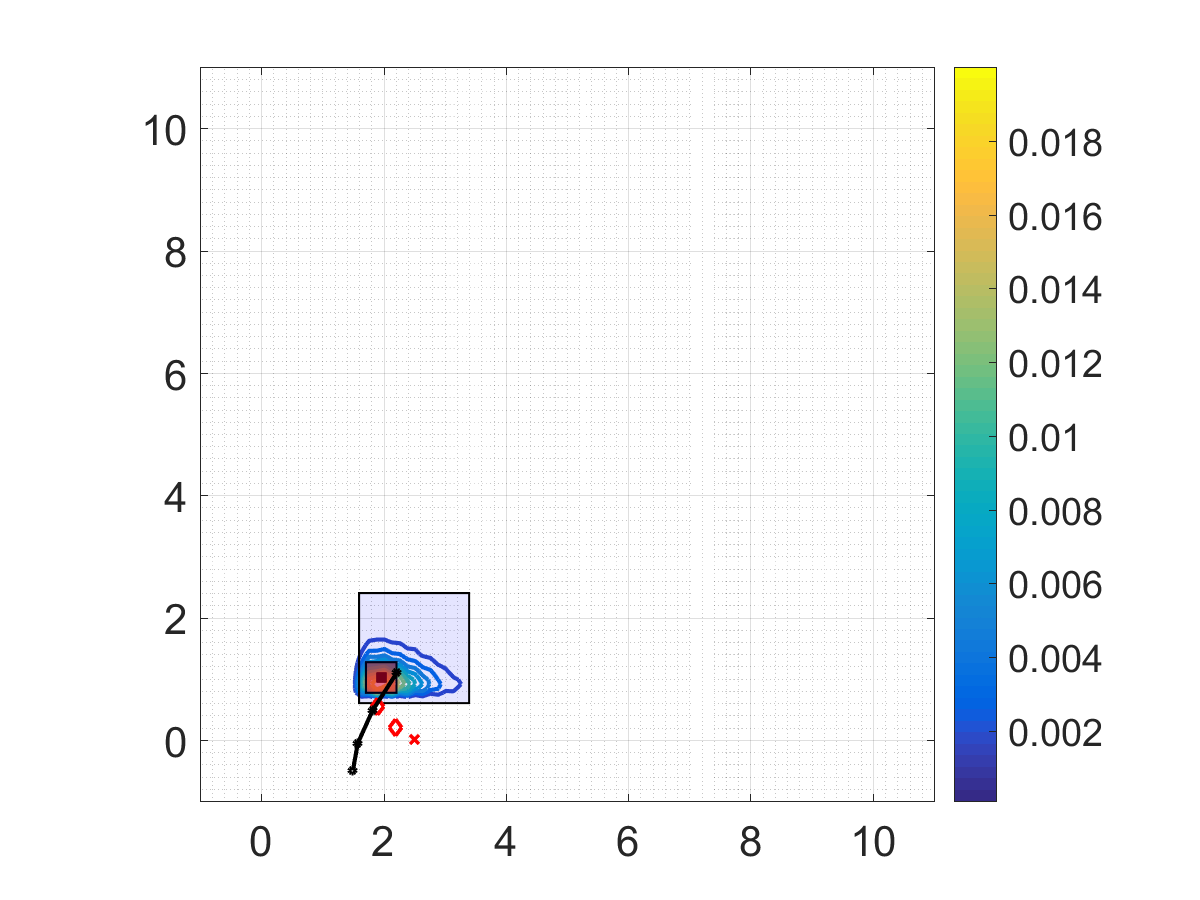} \label{fig:sim23}
        \begin{tikzpicture}[remember picture,overlay] 
            \node [text centered,text=black] at (-4.8em,7em) (A) {\footnotesize $\begin{aligned}
                \mathrm{Time}&=3\\
                \mathrm{CapturePr}_{\bar{x}_R}^\ast&=0.3885
            \end{aligned}$};
    \end{tikzpicture}
    \end{subfigure}    
    ~ 
    \begin{subfigure}{\figwidthSnap}
        \includegraphics[width=1\linewidth]{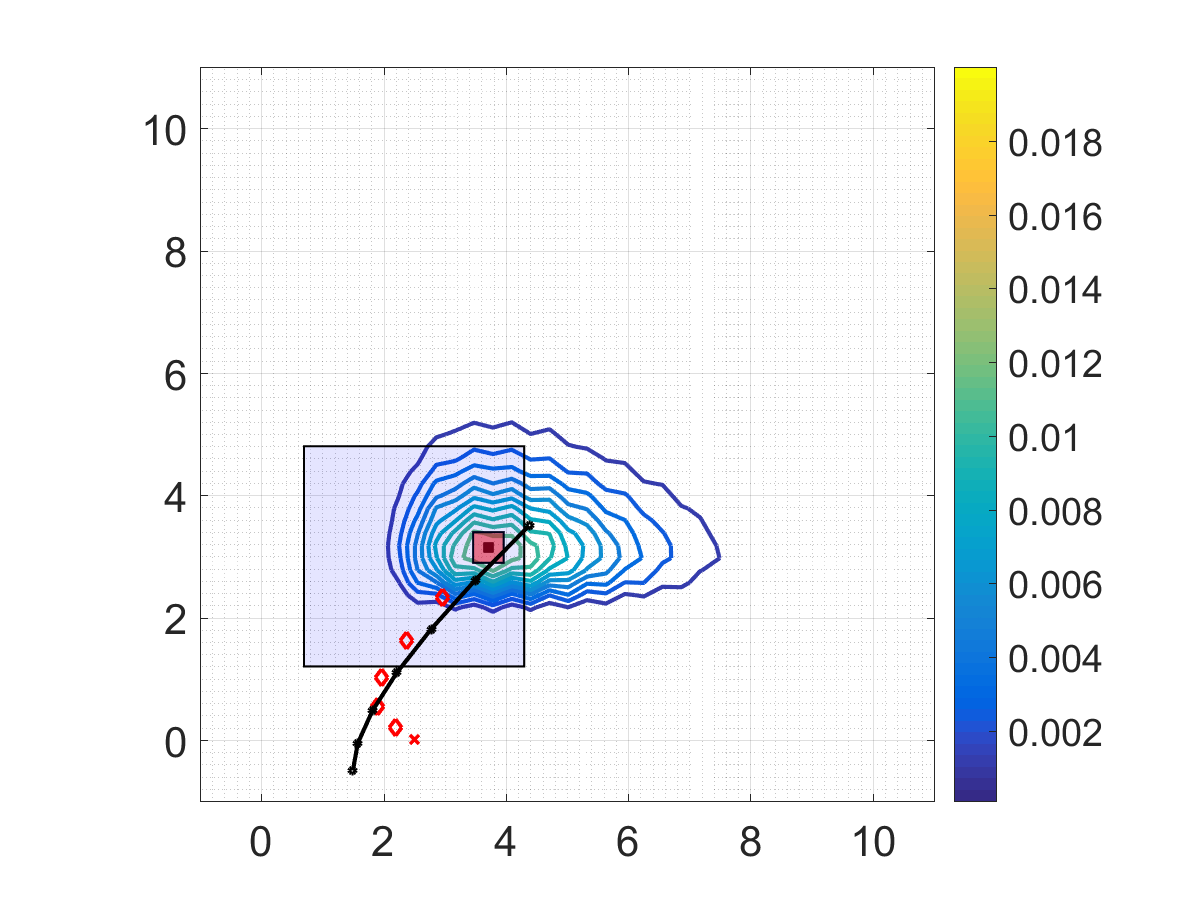}        
        \begin{tikzpicture}[remember picture,overlay] 
            \node [text centered,text=black] at (-4.8em,7em) (A) {\footnotesize $\begin{aligned}
                \mathrm{Time}&=6\\
                \mathrm{CapturePr}_{\bar{x}_R}^\ast&=0.0495
            \end{aligned}$};
    \end{tikzpicture}
    \end{subfigure}
    ~ 
    \begin{subfigure}{\figwidthSnap}
        \includegraphics[width=1\linewidth]{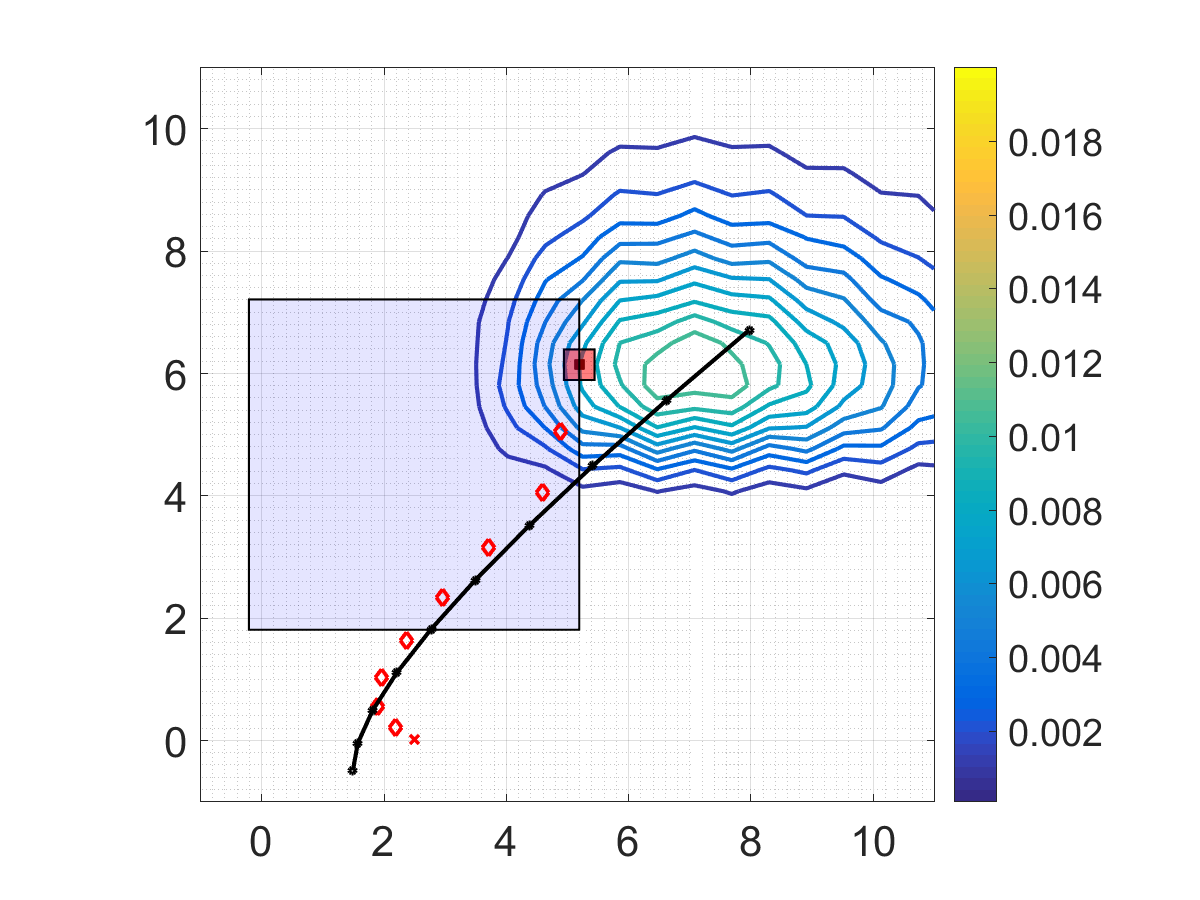} \label{fig:sim29}
        \begin{tikzpicture}[remember picture,overlay] 
            \node [text centered,text=black] at (-4.8em,7em) (A) {\footnotesize $\begin{aligned}
                \mathrm{Time}&=9\\
                \mathrm{CapturePr}_{\bar{x}_R}^\ast&=0.0091
            \end{aligned}$};
    \end{tikzpicture}
     \end{subfigure}    
     \caption{Snapshots of optimal capture positions of the robots G and R when
         G has double integrator dynamics~\eqref{eq:robotG_DI}. The
         blue line shows the mean position trajectory of robot G $\mu_G[\tau]$, the
         contour plot characterizes
         $\psi_{\bx_G}^\mathrm{pos}(\cdot;\tau,\bar{x}_G[0])$ via Monte-Carlo
         simulation,
         the blue box shows the reach set of
         the robot R at time $\tau$, $\mathrm{Reach}_R(\tau,\bar{x}_R[0])$, and the red
         box shows the capture region centered at $\bar{x}_R^\ast[\tau]$,
     $\mathrm{CaptureSet}(\bar{x}_R^\ast[\tau])$.  }\label{fig:SIM2_exp_snapshot}
\end{figure}
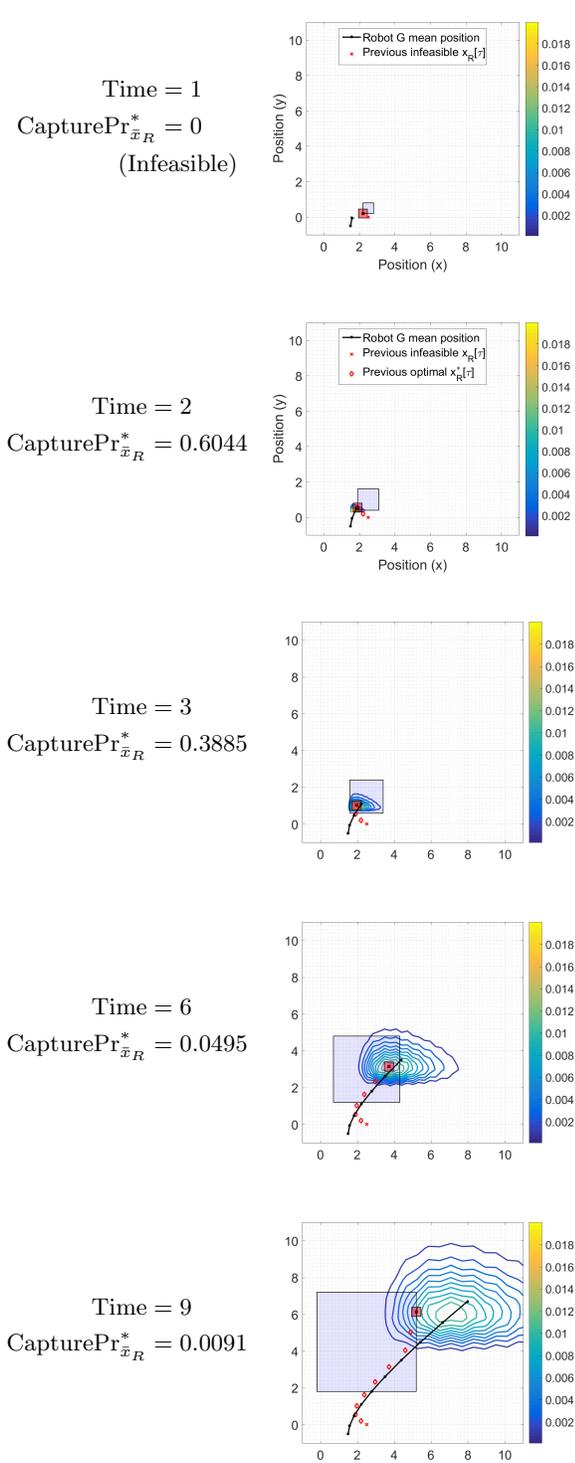

\begin{figure}
    \centering 
    \includegraphics[width=0.8\linewidth]{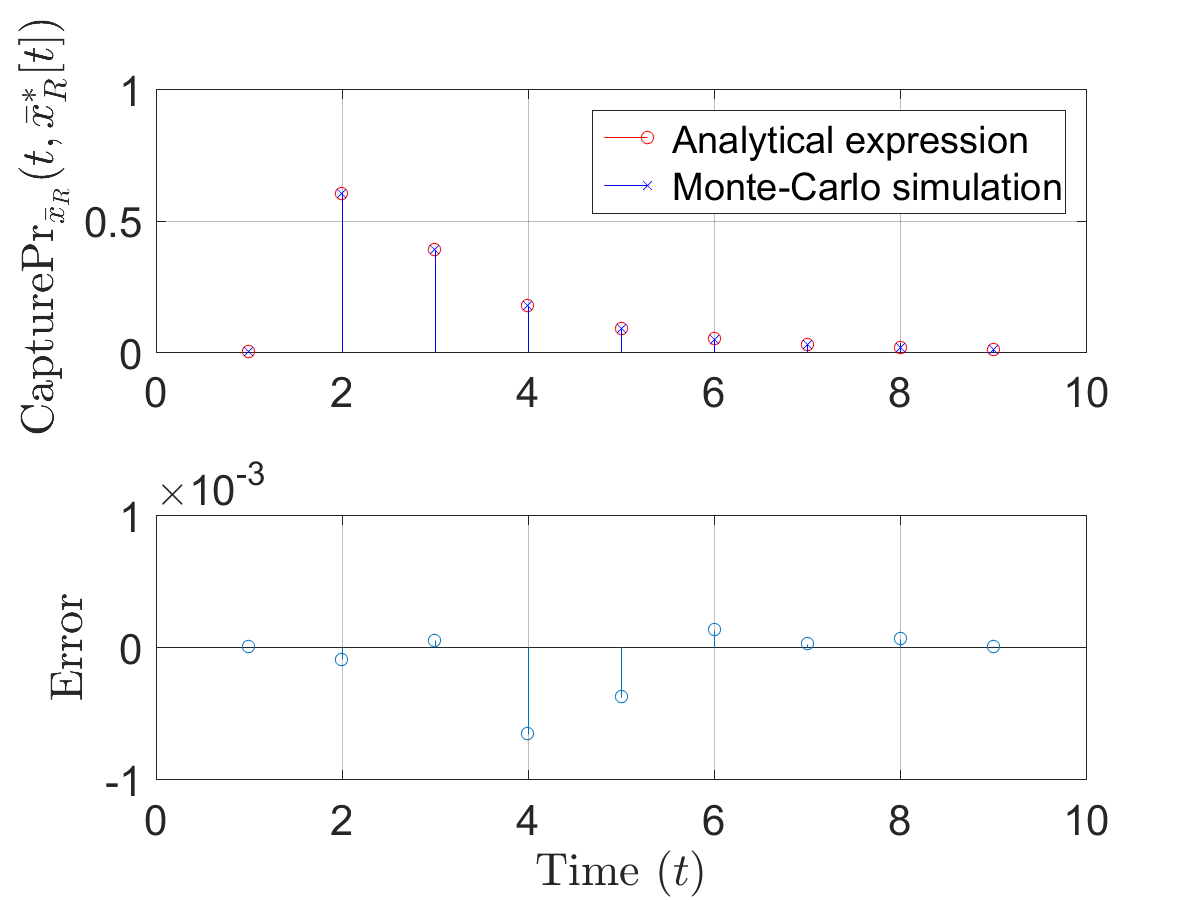} 
    \caption{Solution to Problem \probref{ProbC} for robot G dynamics in
        \eqref{eq:robotG_DI}, and validation of
        $\mathrm{CapturePr}_{\bar{x}_R}(\tau,\bar{x}_R^\ast[\tau];\bar{x}_G[0])$ via
        Monte-Carlo simulations. The optimal capture time is $\tau^\ast=2$ and the capture probability is
$\mathrm{CapturePr}_{\bar{x}_R}(\tau^\ast,\bar{x}_R^\ast[\tau^\ast];\bar{x}_G[0])=0.6044$.}\label{fig:SIM2_exp_time}
\end{figure}

We now consider a more complicated capture problem, in which the disturbance is
exponential (hence tracking the mean has little relevance because it is not the
mode, the global maxima of the density), and the robot dynamics are
more realistic.  We solve Problem \probref{ProbB} for the system given by
\eqref{eq:robotG_DI}.  Here, the disturbance set is $ \mathcal{W}=
\mathbb{R}^2_+$. Based on the mean of the stochastic acceleration $\ba[t]$, the
mean position of robot G has a parabolic trajectory due to the double integrator
dynamics, as opposed to the linear trajectory seen in
Subsection~\ref{sub:gauss_PM}. Also, in this case, we do not have an explicit
expression for the FSRPD like Proposition~\ref{prop:FSRPD_gauss}. Using
Theorem~\ref{thm:FSRPD_def}, we obtain an explicit expression for the CF of the
FSRPD. We utilize Lemma~\ref{lem:capturePr} to evaluate
$\mathrm{CapturePr}(\cdot)$.

Analogous to Lemma~\ref{lem:FSRset_gauss} and
Proposition~\ref{prop:FSRPD_gauss}, we characterize the FSR set in
Lemma~\ref{lem:FSRset_exp} and the FSRPD in Proposition~\ref{prop:FSRPD_exp}.
We use Lemma~\ref{lem:FSRset_oplus} to obtain an overapproximation of the FSR
set due to the unavailability of FSRPD to use \eqref{eq:FSRset}.

\begin{lem}\label{lem:FSRset_exp}
    For the system given in \eqref{eq:robotG_DI} with initial state
    $\bar{x}_G[0] \in \mathbb{R}^4$ of the robot G, we have
    $\FSRset_G(\tau,\bar{x}_G[0])\subseteq \{A_\mathrm{G,DI}^\tau\bar{x}_G[0]\}\oplus
    \mathbb{R}^4_+$ for every $2\leq \tau\leq T$, and\\
    $\FSRset_G(1,\bar{x}_G[0])\subseteq\{A_\mathrm{G,DI}\bar{x}_G[0]\}\oplus
    B_\mathrm{G,DI} \mathbb{R}^2_+$.
\end{lem}
\begin{proof}
For the dynamics in \eqref{eq:robotG_DI},
$\mathscr{C}^\top_{4\times(2{\tau})} \mathbb{R}^{2{\tau}}_+= \mathbb{R}^4_+$ since the rank
of $\mathscr{C}^\top_{4\times(2{\tau})}$ is 4 for every ${\tau}\geq 2$, and elements of
$\mathscr{C}^\top_{4\times(2\tau)}$ are nonnegative. For ${\tau}=1$,
$\mathscr{C}^\top_{4\times(2\tau)}=B_\mathrm{G,DI}$. Lemma~\ref{lem:FSRset_oplus}
completes the proof.
\end{proof}
\ \\    
\begin{prop}\label{prop:FSRPD_exp}
    The CF of the FSRPD of the robot G for
    dynamics \eqref{eq:robotG_DI} is 
    \begin{align}
        \Psi_{\bx_G}(\bar{\beta};\tau,\bar{x}_G[0])&=\mathrm{exp}(j\bar{\beta}^\top
        (A^\tau_\mathrm{G,DI}\bar{x}_G[0]))\times \nonumber \\
        &\ \prod_{t=0}^{\tau-1}\frac{\lambda_\mathrm{ax}\lambda_\mathrm{ay}}{(\lambda_\mathrm{ax}-j\bar{\alpha}_{2t})(\lambda_\mathrm{ay}-j\bar{\alpha}_{2t+1})}\label{eq:FSRPD_linear_exp}
    \end{align}
    where $\bar{\alpha}=\mathscr{C}_{4\times (2\tau)}^\top\bar{\beta}\in
    \mathbb{R}^{(2\tau)}$ and $\bar{\beta}\in \mathbb{R}^{4}$.
    The FSRPD of the robot G is $\psi_{\bx_G}(\bar{x};\tau,\bar{x}_G[0])=
    \mathscr{F}^{-1}\big\{\Psi_{\bx_G}(\cdot;\tau,\bar{x}_G[0])\big\}(-\bar{x})$.
\end{prop}
\begin{proof}
    Apply Theorem~\ref{thm:FSRPD_def} to the dynamics
    \eqref{eq:robotG_DI}.
\end{proof}

To solve Problem~\probref{ProbB}, we define
$\mathrm{CapturePr}_{\bar{x}_R}(\cdot)$ as in \eqref{eq:intprob}. Since we are interested in just the position of robot G, we require only the
marginal density of the FSRPD over the position subspace of robot G,
$\psi_{\bx_G}^\mathrm{pos}$. By Property P4, we have for
$\bar{\gamma}=[\gamma_1\ \gamma_2]\in \mathbb{R}^2$,
\begin{align}
    \Psi_{\bx_G}^\mathrm{pos}(\bar{\gamma};\tau,\bar{x}_G[0])&=\Psi_{\bx_G}({[\gamma_1\
    0\ \gamma_2\ 0]}^\top;\tau,\bar{x}_G[0]).\label{eq:psi_pos_exp}
\end{align}
Unlike the case with Gaussian disturbance, explicit expressions for the FSRPD
$\psi_{\bx_G}$ or its marginal density $\psi_{\bx_G}^\mathrm{pos}$ are
unavailable since the Fourier transform \eqref{eq:FSRPD_linear_exp} is not
standard. 
\begin{lem}
$\psi_{B_\mathrm{G,DI}\ba},\ \psi_{\bx_G}\in L^1 ( \mathbb{R}^4) \cap L^2 (
    \mathbb{R}^4)$.\label{lem:FSRPD_L2}
\end{lem}
\begin{proof}
    (\emph{For $\psi_{B_\mathrm{G,DI}\ba}$})
By H\"{o}lder's inequality\cite[Section 19]{billingsley_probability_1995},
$\psi_{\ba}\in L^1(\mathbb{R}^2) \cap L^2 (\mathbb{R}^2)$. We also have\\
$\psi_{B_\mathrm{G,DI}\ba}(z_1,z_2,z_3,z_4)=\delta(z_3-\frac{T_sz_4}{2})\delta(z_1-\frac{T_sz_2}{2})\psi_{\bz_{24}}(z_2,z_4)$
where $\bz_{24}=[z_2\ z_4]^\top=T_s\ba\in \mathbb{R}^2$ and
$\psi_{\bz_{24}}(z_2,z_4)=T_s^{-2}\psi_{\ba}(\frac{z_2}{T_s},\frac{z_4}{T_s})$
from \eqref{eq:probTransform}. For $i=\{1,2\}$, ${\Vert
\psi_{B_\mathrm{G,DI}\ba}\Vert}_i={\Vert
\psi_{\bz_{24}}\Vert}_i=T_s^{2-2i}{\Vert \psi_{\ba}\Vert}_i<\infty$ completing
the proof.
    
(\emph{For $\psi_{\bx_G}$})
Via induction using \eqref{eq:recurs_FSRPD} (similar to the proof of
Theorem~\ref{thm:FSRPD_log_concave}). Note that functions in $L^1 (
\mathbb{R}^4) \cap L^2( \mathbb{R}^4)$ are closed under
convolution~\cite[Theorem 1.3]{stein1971introduction}.
\end{proof}

\begin{lem}
    $\psi_{\bx_G}^\mathrm{pos}(\bar{x};\tau,\bar{x}_G[0])\in L^1 ( \mathbb{R}^2)
    \cap L^2 ( \mathbb{R}^2)$.\label{lem:marg_dens}
\end{lem}
\begin{proof}
    For $i=\{1,2\}$, we have from \eqref{eq:psi_pos_exp}, ${\Vert
    \psi_{\bx_G}^\mathrm{pos} \Vert}_i={\Vert \psi_{\bx_G}\Vert}_i$, and from
    Lemma~\ref{lem:FSRPD_L2}, ${\Vert \psi_{\bx_G}\Vert}_i <\infty$.
\end{proof}

Similar to Subsection~\ref{sub:gauss_PM}, we define a convex capture region
$\mathrm{CaptureSet}(\bar{y}_R)=\mathrm{Box}(\bar{y}_R,a)\subseteq \mathbb{R}^2$
where $\bar{y}_R\in \mathbb{R}^2$ is the state of the robot R.  We define
$h(\bar{y};\bar{y}_R,a)=\mathbf 1_{\mathrm{Box}(\bar{y}_R,a)}(\bar{y})$ as the
indicator function corresponding to a $2$-D box centered at $\bar{y}_R$ with
edge length $2a>0$ with $h(\bar{y})=1$ if $\bar{y}\in
\mathrm{CaptureSet}(\bar{y}_R)$ and zero otherwise. The Fourier transform of $h$
is a product of \emph{sinc} functions shifted by $\bar{y}_R$ (follows from
Property P2 and~\cite[Chapter 13]{bracewell_fourier_1986})
\begin{align}
    H(\bar{\gamma};\bar{y}_R,a)&=\mathscr{F}\{h(\cdot;\bar{y}_R,a)\}(\bar{\gamma}) \nonumber \\
    &=4a^2\exp{(-j\bar{y}_R^\top\bar{\gamma})}\frac{\sin(a
\gamma_1)\sin(a \gamma_2)}{\gamma_1\gamma_2}.\label{eq:FT_box}
\end{align}
Clearly, $h$ is square-integrable, and from Lemmas~\ref{lem:capturePr}
and~\ref{lem:marg_dens}, we define
$\mathrm{CapturePr}_{\bar{x}_R}(\cdot)$ in \eqref{eq:capturePr_DI}. 
Equation \eqref{eq:capturePr_DI} is
evaluated using \eqref{eq:FSRPD_linear_exp}, \eqref{eq:psi_pos_exp}, and \eqref{eq:FT_box}. We use \eqref{eq:capturePr_DI} as opposed
\eqref{eq:capturePr_DI_density} due to the unavailability of an explicit expression
for $\psi_{\bx_G}^\mathrm{pos}$. The numerical evaluation of the inverse Fourier
transform of $\Psi_{\bx_G}^\mathrm{pos}$ to compute \eqref{eq:capturePr_DI_density}
will require two quadratures, resulting in a higher approximation error as compared
to \eqref{eq:capturePr_DI}.

\begin{figure*}
\begin{align}
    \mathrm{CapturePr}_{\bar{x}_R}(\tau,\bar{x}_R[\tau];\bar{x}_G[0])&=\int_{ \mathbb{R}^2}
\psi_{\bx_G}^\mathrm{pos}(\bar{x};\tau,\bar{x}_G[0])h(\bar{x};\bar{x}_R[\tau],a)d\bar{x}
\label{eq:capturePr_DI_density}\\
&={\left(\frac{1}{2\pi}\right)}^2\int_{ \mathbb{R}^2}
\Psi_{\bx_G}^\mathrm{pos}(\bar{\gamma};\tau,\bar{x}_G[0])H(\bar{\gamma};\bar{x}_R[\tau],a)d\bar{\gamma}.\label{eq:capturePr_DI}
\end{align}\rule{\textwidth}{0.5pt}
\end{figure*}

We implement the problem with the following parameters: $T_s=0.2$, $T=9$,
$a=0.25$, $\lambda_\mathrm{ax}=0.25$, $\lambda_\mathrm{ay}=0.45$, $\bar{x}_G[0]=
{[1.5\ 0\ -0.5\ 2]}^\top$, $\bar{x}_R[0]=[2.5\ 0]^\top$, and $
\mathcal{U}=[-1.5,1.5]\times[1,4]$. 
We use $J_\pi(\bar{\pi})=0$ in Problem~\probref{ProbD}.

Figure~\ref{fig:SIM2_exp_snapshot} shows the evolution of the mean position of
the robot G and the optimal capture position for the robot R at time instants
$1,2,3,6,$ and $9$. For every $\tau\in[1,T]$, the contour plots of
$\psi_{\bx_G}^\mathrm{pos}(\cdot;\tau,\bar{x}_G[0])$ were estimated via Monte-Carlo
simulation 
since evaluating $\psi_{\bx_G}^\mathrm{pos}(\cdot;\tau,\bar{x}_G[0])$ via
\eqref{eq:cfun_ift} over a grid is computationally expensive. Note that the mean
position of the robot G does not coincide with the mode of
$\psi_{\bx_G}^\mathrm{pos}(\cdot;\tau,\bar{x}_G[0])$ in contrast to the problem
discussed in Subsection~\ref{sub:gauss_PM}.  The optimal time of capture is at
$\tau^\ast=2$, the optimal capture position is $\bar{x}_R^\ast[\tau^\ast]=[1.9\
0.55]^\top$, and the corresponding probability of robot R capturing robot G is
$0.6044$ (Figure~\ref{fig:SIM2_exp_snapshot}b).
Figure~\ref{fig:SIM2_exp_time} shows the optimal capture probabilities obtained
when solving Problem \probref{ProbC} for the dynamics~\eqref{eq:robotG_DI}, and
the validation of the results. 

\subsection{Numerical implementation and analysis}
\label{sub:numImp}

All computations in this paper were performed using MATLAB on an Intel Core i7
CPU with 3.4GHz clock rate and 16 GB RAM. The MATLAB code for this work is
available at \url{http://hscl.unm.edu/files/code/HSCC17.zip}.

We solved Problem~\probref{ProbC} using MATLAB's built-in functions ---
\emph{fmincon} for the optimization, \emph{mvncdf} to compute the objective
\eqref{eq:intprob} for the case in Subsection~\ref{sub:gauss_PM},
\emph{integral} to compute the objective \eqref{eq:capturePr_DI} for the case in
Subsection~\ref{sub:exp_DI}, and \emph{max} to compute the global optimum of
Problem~\probref{ProbB}. In both the sections, we used MPT for the reachable set
calculation and solved Problem~\probref{ProbD} using CVX~\cite{cvx}. Using
Lemma~\ref{lem:FSRmembership}, the FSR sets 
restrict the
search while solving Problem ProbC. All geometric computations were done in
the facet representation. We computed the initial guess for the optimization of
Problem~\probref{ProbC} by performing Euclidean projection of the mean to the feasible set
using CVX~\cite[Section 8.1.1]{boyd_convex_2004}. Since computing the objective was costly, this operation saved
significant computational time.  The Monte-Carlo simulation 
used $500,000$ particles. No offline computations were
done in either of the cases.

The overall computation of Problem \probref{ProbB} and \probref{ProbD} for the
case in Subsection~\ref{sub:gauss_PM} took $5.32$ seconds for $T=20$. Since
Proposition~\ref{prop:FSRPD_gauss} provides explicit expressions for the FSRPD,
the evaluation of the FSRPD for any given point $\bar{y}\in \mathcal{X}$ takes
$1.6$ millseconds on average. For the case in Subsection~\ref{sub:exp_DI}, the
overall computation took $488.55$ seconds ($\sim 8$ minutes) for $T=9$. The
numerical evaluation of the improper integral \eqref{eq:capturePr_DI} is the
major cause of increase in runtime.  The evaluation of the FSRPD for any given
point $\bar{y}\in \mathcal{X}$ using \eqref{eq:cfun_ift} takes about $10.5$
seconds, and the runtime and the accuracy depend heavily on the point $\bar{y}$
as well as the bounds used for the integral approximation. However, the
evaluation of $\mathrm{CapturePr}_{\bar{x}_R}(\cdot)$ using \eqref{eq:capturePr_DI} is much
faster ($0.81$ seconds) because $H(\bar{\gamma};\bar{y}_R,a)$ is a decaying, 2-D
sinc function (decaying much faster than the CF).  

The decaying properties of the integrand in \eqref{eq:capturePr_DI} and CFs in
general permits approximating the improper integrals in \eqref{eq:cfun_ift}
and \eqref{eq:capturePr_DI} by as a proper integral with suitably defined finite
bounds. The tradeoff between accuracy and computational speed, common in
quadrature techniques, dictates the choice of the bound. A detailed analysis of
various quadrature techniques, their computational complexity, and their error
analysis can be found in~\cite[Chapter 4]{press2007numerical}.

\section{Conclusions and Future work}
\label{sec:conc}

This paper provides a method for forward stochastic reachability analysis using
Fourier transforms. The method is applicable to uncontrolled stochastic linear
systems. Fourier transforms simplify the computation and mitigate the curse of
dimensionality associated with gridding the state space. We also analyze several
convexity results associated with the FSRPD and FSR sets. We demonstrate our
method on the problem of controller synthesis for a controlled robot pursuing a stochastically moving non-adversarial target.

Future work includes exploration of various quadrature techniques like particle
filters for high-dimensional quadratures, and extension to a model predictive
control framework and to discrete random vectors (countable disturbance sets).
Multiple pursuer applications will also be investigated.

\section{Acknowledgements}
\label{sec:acknowledgements}

The authors thank Prof. M. Hayat for discussions on Fourier transforms in
probability theory and the reviewers for their insightful comments. 
 
This material is based upon work supported by the National Science Foundation,
under Grant Numbers CMMI-1254990, CNS-1329878, and IIS-1528047.
Any opinions, findings, and conclusions or recommendations expressed in this
material are those of the authors and do not necessarily reflect the views of
the National Science Foundation.


\bibliographystyle{unsrt}
\bibliography{IEEEabrv,shortIEEE,obstacleavoid}
\end{document}